%% file: main.tex
\documentclass[a4paper,fleqn]{cas-dc}

\usepackage{lmodern}

\pdftrailerid{}
\usepackage[numbers]{natbib}
\usepackage{amsmath,amssymb,amsthm}

\makeatletter
\@ifpackageloaded{stix}{%
  \let\mathrm\relax     \DeclareSymbolFontAlphabet{\mathrm}    {operators}%
  \let\mathnormal\relax \DeclareSymbolFontAlphabet{\mathnormal}{letters}%
  \let\mathcal\relax    \DeclareSymbolFontAlphabet{\mathcal}   {symbols}%
  \let\mathbb\relax     \DeclareSymbolFontAlphabet{\mathbb}    {AMSb}%
  \let\mathbf\relax     \DeclareMathAlphabet      {\mathbf}    {OT1}{lmr} {bx}{n}%
  \let\mathsf\relax     \DeclareMathAlphabet      {\mathsf}    {OT1}{lmss}{m}{n}%
  \let\mathit\relax     \DeclareMathAlphabet      {\mathit}    {OT1}{lmr} {m}{it}%
  \let\mathtt\relax     \DeclareMathAlphabet      {\mathtt}    {OT1}{lmtt}{m}{n}%
  \SetMathAlphabet{\mathbf}{bold}{OT1}{lmr} {bx}{n}%
  \SetMathAlphabet{\mathsf}{bold}{OT1}{lmss}{bx}{n}%
  \SetMathAlphabet{\mathit}{bold}{OT1}{lmr} {bx}{it}%
  \SetMathAlphabet{\mathtt}{bold}{OT1}{lmtt}{m}{n}%
  \input{xir-math-symbols}}{}
\makeatother

\usepackage{graphicx}
\usepackage{booktabs}
\usepackage{array}
\usepackage{float}
\usepackage{placeins}
\usepackage{xcolor}
\usepackage{listings}
\usepackage{algorithm}
\usepackage{algorithmicx}
\usepackage{algpseudocode}

\usepackage{tikz}
\usetikzlibrary{arrows.meta,positioning,fit,backgrounds,calc,bending}
\input{figures/xir-icons}
\usepackage{hyperref}
\usepackage{url}
\usepackage{orcidlink}
\hypersetup{
  hidelinks,
  pdftitle={XIR: A Framework for Interoperability across Cross-Chain Protocols Based on a Verifiable Intermediate Representation},
  pdfauthor={Yushen Li, Linpeng Jia, Jiaying Feng, Ziliang Liao, Yi Sun},
  pdfsubject={Cross-chain protocol interoperability and reachability},
  pdfkeywords={Blockchain, Cross-chain, Interoperability, Reachability}
}
\input{generated/paper-values}

\theoremstyle{definition}

\newtheorem{lemma}{Lemma}

\newtheorem{theorem}{Theorem}

\ExplSyntaxOn
\RenewDocumentCommand \emailauthor { m m }
  {
    \int_gincr:N \g_ead_int
    \seq_gput_right:Nn \g_stm_ead_seq
      {
        \mbox{
          \normalfont\mdseries
          \href{mailto:#1}{\textcolor[HTML]{007A99}{\tl_to_str:n {#1}}}
          \parsename {#2}
          \space(\eadauthor)
        }
      }
  }
\RenewDocumentCommand \printemails { }
  {
    \group_begin:
    \int_compare:nNnT {\g_ead_int} > {0}
      {
        \tex_let:D \thefootnote \relax
        \footnotetext{
          \raggedright\normalfont\mdseries
          \mbox{\textit{E-mail~address:}~\seq_use:Nn \g_stm_ead_seq {;~}.}
        }
      }
    \group_end:
  }
\ExplSyntaxOff

\begin{document}
\let\WriteBookmarks\relax
\def\floatpagepagefraction{1}
\def\textpagefraction{.001}

\shorttitle{XIR}
\shortauthors{Y. Li et~al.}

\ExplSyntaxOn
\cs_set:Npn \__first_footerline:
  {
    \group_begin:
    \small
    \sffamily
    \__short_authors: :~
    { \rmfamily \itshape Accepted~for~publication~in~Blockchain:~Research~and~Applications }
    \group_end:
  }
\cs_set:Npn \__cas_head:
  {
    \parbox{\textwidth}
      {
        \sffamily\small
        Y.~Li,~L.~Jia,~J.~Feng,~Z.~Liao~and~Y.~Sun
        \hfill
        Blockchain:~Research~and~Applications
      }
  }
\makeatletter\ps@cas\makeatother
\ExplSyntaxOff

\title[mode = title]{XIR: A Framework for Interoperability across Cross-Chain
Protocols Based on a Verifiable Intermediate Representation}

\author[1,2]{Yushen Li}[orcid=0009-0003-5069-0337,suffix={\orcidlink{0009-0003-5069-0337}}]
\author[1]{Linpeng Jia}[orcid=0000-0003-1916-6193,suffix={\orcidlink{0000-0003-1916-6193}}]
\cormark[1]
\ead{jialinpeng@ict.ac.cn}
\author[1,3]{Jiaying Feng}[orcid=0009-0009-3418-1552,suffix={\orcidlink{0009-0009-3418-1552}}]
\author[1,3]{Ziliang Liao}[orcid=0009-0008-0259-243X,suffix={\orcidlink{0009-0008-0259-243X}}]
\author[1,2,3,4,5]{Yi Sun}[orcid=0000-0002-9817-8004,suffix={\orcidlink{0000-0002-9817-8004}}]
\cortext[cor1]{Corresponding author.}

\affiliation[1]{
  organization={Institute of Computing Technology, Chinese Academy of Sciences},
  city={Beijing},
  postcode={100190},
  country={China}}
\affiliation[2]{
  organization={Hangzhou Institute for Advanced Study, University of Chinese Academy of Sciences},
  city={Hangzhou},
  postcode={310024},
  country={China}}
\affiliation[3]{
  organization={School of Computer Science and Technology, University of Chinese Academy of Sciences},
  city={Beijing},
  postcode={100049},
  country={China}}
\affiliation[4]{
  organization={Beijing Advanced Innovation Center for Future Blockchain and Privacy Computing, Beihang University},
  city={Beijing},
  postcode={100191},
  country={China}}
\affiliation[5]{
  organization={Shandong Key Laboratory of Blockchain Finance, Shandong University of Finance and Economics},
  city={Jinan},
  postcode={250014},
  country={China}}

\begin{abstract}
Cross-chain protocols enable applications to exchange messages across
blockchains. Under point-to-point configurations, communication depends on a
direct connection between the source and destination blockchains, limiting
blockchain reachability and requiring additional configurations to connect more
blockchains. To quantify this problem, this paper analyzes approximately 25 million
mainnet cross-chain transaction events collected from six protocols (Axelar, CCIP,
Hyperlane, LayerZero, Relay, and Wormhole) between January and October 2025\footnotemark[1]. The
resulting graph covers 286 active blockchains and 11,935 directly connected ordered
blockchain pairs. These connections provide a direct reachability of 14.64\%, while
full direct connectivity would require 81,510 point-to-point configurations.
We present XIR, a framework for interoperability across cross-chain protocols based
on a verifiable intermediate representation. This representation binds an application
message to an ordered record of authenticated cross-chain protocol deliveries,
preserving message identity and verification history across protocol boundaries.
XIR Gateways and XIR Adapters use this representation to compose
existing connections into same-protocol and cross-protocol multi-hop paths. We
implement an XIR prototype integrating Hyperlane and LayerZero and evaluate it in
local and public-testnet environments\footnotemark[2]. Theoretical analysis and evaluation show that, with correctly configured cross-chain protocol connections, XIR avoids 67,018 additional point-to-point configurations, equivalent to 84.88\% of the total required by a point-to-point configuration baseline serving the same reachable pairs, and increases reachability from 14.64\% to 96.86\% of all ordered
blockchain pairs.
\end{abstract}


\maketitle
\footnotetext[1]{XIR cross-chain event dataset: \url{https://doi.org/10.34740/kaggle/dsv/19934883}.}
\footnotetext[2]{XIR implementation: \url{https://github.com/lysrain21/XIR}.}
\setcounter{footnote}{2}
\hypersetup{
  pdfauthor={Yushen Li, Linpeng Jia, Jiaying Feng, Ziliang Liao, Yi Sun},
  pdfsubject={Cross-chain protocol interoperability and reachability}
}

\input{sections/1-introduction}
\input{sections/2-background}
\input{sections/3-system-model}
\input{sections/4-xir-design}
\input{sections/5-analysis}
\input{sections/6-implementation-evaluation}

\input{sections/7-discussion}
\input{sections/8-conclusion}

\input{sections/submission-declarations}

\setlength{\bibsep}{0pt}
\renewcommand{\bibfont}{\fontsize{8pt}{9pt}\selectfont}
\bibliographystyle{cas-model2-names}
\bibliography{refs}

\end{document}

%% file: xir-math-symbols.tex
%

\DeclareMathSymbol{a}{\mathalpha}{letters}{`a}
\DeclareMathDelimiter{(}{\mathopen}{operators}{"28}{largesymbols}{"00}
\DeclareMathDelimiter{[}{\mathopen}{operators}{"5B}{largesymbols}{"02}
\expandafter\let\csname lmoustache\endcsname\relax
\DeclareMathDelimiter{\lmoustache}{\mathopen}{largesymbols}{"7A}{largesymbols}{"40}
\expandafter\let\csname rmoustache\endcsname\relax
\DeclareMathDelimiter{\rmoustache}{\mathclose}{largesymbols}{"7B}{largesymbols}{"41}
\expandafter\let\csname arrowvert\endcsname\relax
\DeclareMathDelimiter{\arrowvert}{\mathord}{symbols}{"6A}{largesymbols}{"3C}
\expandafter\let\csname Arrowvert\endcsname\relax
\DeclareMathDelimiter{\Arrowvert}{\mathord}{symbols}{"6B}{largesymbols}{"3D}
\expandafter\let\csname Vert\endcsname\relax
\DeclareMathDelimiter{\Vert}{\mathord}{symbols}{"6B}{largesymbols}{"0D}
\DeclareMathDelimiter{\|}{\mathord}{symbols}{"6B}{largesymbols}{"0D}
\expandafter\let\csname vert\endcsname\relax
\DeclareMathDelimiter{\vert}{\mathord}{symbols}{"6A}{largesymbols}{"0C}
\expandafter\let\csname uparrow\endcsname\relax
\DeclareMathDelimiter{\uparrow}{\mathrel}{symbols}{"22}{largesymbols}{"78}
\expandafter\let\csname downarrow\endcsname\relax
\DeclareMathDelimiter{\downarrow}{\mathrel}{symbols}{"23}{largesymbols}{"79}
\expandafter\let\csname updownarrow\endcsname\relax
\DeclareMathDelimiter{\updownarrow}{\mathrel}{symbols}{"6C}{largesymbols}{"3F}
\expandafter\let\csname Uparrow\endcsname\relax
\DeclareMathDelimiter{\Uparrow}{\mathrel}{symbols}{"2A}{largesymbols}{"7E}
\expandafter\let\csname Downarrow\endcsname\relax
\DeclareMathDelimiter{\Downarrow}{\mathrel}{symbols}{"2B}{largesymbols}{"7F}
\expandafter\let\csname Updownarrow\endcsname\relax
\DeclareMathDelimiter{\Updownarrow}{\mathrel}{symbols}{"6D}{largesymbols}{"77}
\expandafter\let\csname backslash\endcsname\relax
\DeclareMathDelimiter{\backslash}{\mathord}{symbols}{"6E}{largesymbols}{"0F}
\expandafter\let\csname rangle\endcsname\relax
\DeclareMathDelimiter{\rangle}{\mathclose}{symbols}{"69}{largesymbols}{"0B}
\expandafter\let\csname langle\endcsname\relax
\DeclareMathDelimiter{\langle}{\mathopen}{symbols}{"68}{largesymbols}{"0A}
\expandafter\let\csname rbrace\endcsname\relax
\DeclareMathDelimiter{\rbrace}{\mathclose}{symbols}{"67}{largesymbols}{"09}
\expandafter\let\csname lbrace\endcsname\relax
\DeclareMathDelimiter{\lbrace}{\mathopen}{symbols}{"66}{largesymbols}{"08}
\expandafter\let\csname rceil\endcsname\relax
\DeclareMathDelimiter{\rceil}{\mathclose}{symbols}{"65}{largesymbols}{"07}
\expandafter\let\csname lceil\endcsname\relax
\DeclareMathDelimiter{\lceil}{\mathopen}{symbols}{"64}{largesymbols}{"06}
\expandafter\let\csname rfloor\endcsname\relax
\DeclareMathDelimiter{\rfloor}{\mathclose}{symbols}{"63}{largesymbols}{"05}
\expandafter\let\csname lfloor\endcsname\relax
\DeclareMathDelimiter{\lfloor}{\mathopen}{symbols}{"62}{largesymbols}{"04}
\expandafter\let\csname lgroup\endcsname\relax
\DeclareMathDelimiter{\lgroup}{\mathopen}{largesymbols}{"3A}{largesymbols}{"3A}
\expandafter\let\csname rgroup\endcsname\relax
\DeclareMathDelimiter{\rgroup}{\mathclose}{largesymbols}{"3B}{largesymbols}{"3B}
\expandafter\let\csname bracevert\endcsname\relax
\DeclareMathDelimiter{\bracevert}{\mathord}{largesymbols}{"3E}{largesymbols}{"3E}
\DeclareMathSymbol{a}{\mathalpha}{letters}{`a}
\DeclareMathSymbol{b}{\mathalpha}{letters}{`b}
\DeclareMathSymbol{c}{\mathalpha}{letters}{`c}
\DeclareMathSymbol{d}{\mathalpha}{letters}{`d}
\DeclareMathSymbol{e}{\mathalpha}{letters}{`e}
\DeclareMathSymbol{f}{\mathalpha}{letters}{`f}
\DeclareMathSymbol{g}{\mathalpha}{letters}{`g}
\DeclareMathSymbol{h}{\mathalpha}{letters}{`h}
\DeclareMathSymbol{i}{\mathalpha}{letters}{`i}
\DeclareMathSymbol{j}{\mathalpha}{letters}{`j}
\DeclareMathSymbol{k}{\mathalpha}{letters}{`k}
\DeclareMathSymbol{l}{\mathalpha}{letters}{`l}
\DeclareMathSymbol{m}{\mathalpha}{letters}{`m}
\DeclareMathSymbol{n}{\mathalpha}{letters}{`n}
\DeclareMathSymbol{o}{\mathalpha}{letters}{`o}
\DeclareMathSymbol{p}{\mathalpha}{letters}{`p}
\DeclareMathSymbol{q}{\mathalpha}{letters}{`q}
\DeclareMathSymbol{r}{\mathalpha}{letters}{`r}
\DeclareMathSymbol{s}{\mathalpha}{letters}{`s}
\DeclareMathSymbol{t}{\mathalpha}{letters}{`t}
\DeclareMathSymbol{u}{\mathalpha}{letters}{`u}
\DeclareMathSymbol{v}{\mathalpha}{letters}{`v}
\DeclareMathSymbol{w}{\mathalpha}{letters}{`w}
\DeclareMathSymbol{x}{\mathalpha}{letters}{`x}
\DeclareMathSymbol{y}{\mathalpha}{letters}{`y}
\DeclareMathSymbol{z}{\mathalpha}{letters}{`z}
\DeclareMathSymbol{A}{\mathalpha}{letters}{`A}
\DeclareMathSymbol{B}{\mathalpha}{letters}{`B}
\DeclareMathSymbol{C}{\mathalpha}{letters}{`C}
\DeclareMathSymbol{D}{\mathalpha}{letters}{`D}
\DeclareMathSymbol{E}{\mathalpha}{letters}{`E}
\DeclareMathSymbol{F}{\mathalpha}{letters}{`F}
\DeclareMathSymbol{G}{\mathalpha}{letters}{`G}
\DeclareMathSymbol{H}{\mathalpha}{letters}{`H}
\DeclareMathSymbol{I}{\mathalpha}{letters}{`I}
\DeclareMathSymbol{J}{\mathalpha}{letters}{`J}
\DeclareMathSymbol{K}{\mathalpha}{letters}{`K}
\DeclareMathSymbol{L}{\mathalpha}{letters}{`L}
\DeclareMathSymbol{M}{\mathalpha}{letters}{`M}
\DeclareMathSymbol{N}{\mathalpha}{letters}{`N}
\DeclareMathSymbol{O}{\mathalpha}{letters}{`O}
\DeclareMathSymbol{P}{\mathalpha}{letters}{`P}
\DeclareMathSymbol{Q}{\mathalpha}{letters}{`Q}
\DeclareMathSymbol{R}{\mathalpha}{letters}{`R}
\DeclareMathSymbol{S}{\mathalpha}{letters}{`S}
\DeclareMathSymbol{T}{\mathalpha}{letters}{`T}
\DeclareMathSymbol{U}{\mathalpha}{letters}{`U}
\DeclareMathSymbol{V}{\mathalpha}{letters}{`V}
\DeclareMathSymbol{W}{\mathalpha}{letters}{`W}
\DeclareMathSymbol{X}{\mathalpha}{letters}{`X}
\DeclareMathSymbol{Y}{\mathalpha}{letters}{`Y}
\DeclareMathSymbol{Z}{\mathalpha}{letters}{`Z}
\DeclareMathSymbol{0}{\mathalpha}{operators}{`0}
\DeclareMathSymbol{1}{\mathalpha}{operators}{`1}
\DeclareMathSymbol{2}{\mathalpha}{operators}{`2}
\DeclareMathSymbol{3}{\mathalpha}{operators}{`3}
\DeclareMathSymbol{4}{\mathalpha}{operators}{`4}
\DeclareMathSymbol{5}{\mathalpha}{operators}{`5}
\DeclareMathSymbol{6}{\mathalpha}{operators}{`6}
\DeclareMathSymbol{7}{\mathalpha}{operators}{`7}
\DeclareMathSymbol{8}{\mathalpha}{operators}{`8}
\DeclareMathSymbol{9}{\mathalpha}{operators}{`9}
\DeclareMathSymbol{!}{\mathclose}{operators}{"21}
\DeclareMathSymbol{*}{\mathbin}{symbols}{"03}
\DeclareMathSymbol{+}{\mathbin}{operators}{"2B}
\DeclareMathSymbol{,}{\mathpunct}{letters}{"3B}
\DeclareMathSymbol{-}{\mathbin}{symbols}{"00}
\DeclareMathSymbol{.}{\mathord}{letters}{"3A}
\DeclareMathSymbol{:}{\mathrel}{operators}{"3A}
\DeclareMathSymbol{;}{\mathpunct}{operators}{"3B}
\DeclareMathSymbol{=}{\mathrel}{operators}{"3D}
\DeclareMathSymbol{?}{\mathclose}{operators}{"3F}
\DeclareMathDelimiter{)}{\mathclose}{operators}{"29}{largesymbols}{"01}
\DeclareMathDelimiter{]}{\mathclose}{operators}{"5D}{largesymbols}{"03}
\DeclareMathDelimiter{<}{\mathopen}{symbols}{"68}{largesymbols}{"0A}
\DeclareMathDelimiter{>}{\mathclose}{symbols}{"69}{largesymbols}{"0B}
\DeclareMathSymbol{<}{\mathrel}{letters}{"3C}
\DeclareMathSymbol{>}{\mathrel}{letters}{"3E}
\DeclareMathDelimiter{/}{\mathord}{operators}{"2F}{largesymbols}{"0E}
\DeclareMathSymbol{/}{\mathord}{letters}{"3D}
\DeclareMathDelimiter{|}{\mathord}{symbols}{"6A}{largesymbols}{"0C}
\expandafter\let\csname alpha\endcsname\relax
\DeclareMathSymbol{\alpha}{\mathord}{letters}{"0B}
\expandafter\let\csname beta\endcsname\relax
\DeclareMathSymbol{\beta}{\mathord}{letters}{"0C}
\expandafter\let\csname gamma\endcsname\relax
\DeclareMathSymbol{\gamma}{\mathord}{letters}{"0D}
\expandafter\let\csname delta\endcsname\relax
\DeclareMathSymbol{\delta}{\mathord}{letters}{"0E}
\expandafter\let\csname epsilon\endcsname\relax
\DeclareMathSymbol{\epsilon}{\mathord}{letters}{"0F}
\expandafter\let\csname zeta\endcsname\relax
\DeclareMathSymbol{\zeta}{\mathord}{letters}{"10}
\expandafter\let\csname eta\endcsname\relax
\DeclareMathSymbol{\eta}{\mathord}{letters}{"11}
\expandafter\let\csname theta\endcsname\relax
\DeclareMathSymbol{\theta}{\mathord}{letters}{"12}
\expandafter\let\csname iota\endcsname\relax
\DeclareMathSymbol{\iota}{\mathord}{letters}{"13}
\expandafter\let\csname kappa\endcsname\relax
\DeclareMathSymbol{\kappa}{\mathord}{letters}{"14}
\expandafter\let\csname lambda\endcsname\relax
\DeclareMathSymbol{\lambda}{\mathord}{letters}{"15}
\expandafter\let\csname mu\endcsname\relax
\DeclareMathSymbol{\mu}{\mathord}{letters}{"16}
\expandafter\let\csname nu\endcsname\relax
\DeclareMathSymbol{\nu}{\mathord}{letters}{"17}
\expandafter\let\csname xi\endcsname\relax
\DeclareMathSymbol{\xi}{\mathord}{letters}{"18}
\expandafter\let\csname pi\endcsname\relax
\DeclareMathSymbol{\pi}{\mathord}{letters}{"19}
\expandafter\let\csname rho\endcsname\relax
\DeclareMathSymbol{\rho}{\mathord}{letters}{"1A}
\expandafter\let\csname sigma\endcsname\relax
\DeclareMathSymbol{\sigma}{\mathord}{letters}{"1B}
\expandafter\let\csname tau\endcsname\relax
\DeclareMathSymbol{\tau}{\mathord}{letters}{"1C}
\expandafter\let\csname upsilon\endcsname\relax
\DeclareMathSymbol{\upsilon}{\mathord}{letters}{"1D}
\expandafter\let\csname phi\endcsname\relax
\DeclareMathSymbol{\phi}{\mathord}{letters}{"1E}
\expandafter\let\csname chi\endcsname\relax
\DeclareMathSymbol{\chi}{\mathord}{letters}{"1F}
\expandafter\let\csname psi\endcsname\relax
\DeclareMathSymbol{\psi}{\mathord}{letters}{"20}
\expandafter\let\csname omega\endcsname\relax
\DeclareMathSymbol{\omega}{\mathord}{letters}{"21}
\expandafter\let\csname varepsilon\endcsname\relax
\DeclareMathSymbol{\varepsilon}{\mathord}{letters}{"22}
\expandafter\let\csname vartheta\endcsname\relax
\DeclareMathSymbol{\vartheta}{\mathord}{letters}{"23}
\expandafter\let\csname varpi\endcsname\relax
\DeclareMathSymbol{\varpi}{\mathord}{letters}{"24}
\expandafter\let\csname varrho\endcsname\relax
\DeclareMathSymbol{\varrho}{\mathord}{letters}{"25}
\expandafter\let\csname varsigma\endcsname\relax
\DeclareMathSymbol{\varsigma}{\mathord}{letters}{"26}
\expandafter\let\csname varphi\endcsname\relax
\DeclareMathSymbol{\varphi}{\mathord}{letters}{"27}
\expandafter\let\csname Gamma\endcsname\relax
\DeclareMathSymbol{\Gamma}{\mathalpha}{operators}{"00}
\expandafter\let\csname Delta\endcsname\relax
\DeclareMathSymbol{\Delta}{\mathalpha}{operators}{"01}
\expandafter\let\csname Theta\endcsname\relax
\DeclareMathSymbol{\Theta}{\mathalpha}{operators}{"02}
\expandafter\let\csname Lambda\endcsname\relax
\DeclareMathSymbol{\Lambda}{\mathalpha}{operators}{"03}
\expandafter\let\csname Xi\endcsname\relax
\DeclareMathSymbol{\Xi}{\mathalpha}{operators}{"04}
\expandafter\let\csname Pi\endcsname\relax
\DeclareMathSymbol{\Pi}{\mathalpha}{operators}{"05}
\expandafter\let\csname Sigma\endcsname\relax
\DeclareMathSymbol{\Sigma}{\mathalpha}{operators}{"06}
\expandafter\let\csname Upsilon\endcsname\relax
\DeclareMathSymbol{\Upsilon}{\mathalpha}{operators}{"07}
\expandafter\let\csname Phi\endcsname\relax
\DeclareMathSymbol{\Phi}{\mathalpha}{operators}{"08}
\expandafter\let\csname Psi\endcsname\relax
\DeclareMathSymbol{\Psi}{\mathalpha}{operators}{"09}
\expandafter\let\csname Omega\endcsname\relax
\DeclareMathSymbol{\Omega}{\mathalpha}{operators}{"0A}
\expandafter\let\csname aleph\endcsname\relax
\DeclareMathSymbol{\aleph}{\mathord}{symbols}{"40}
\expandafter\let\csname imath\endcsname\relax
\DeclareMathSymbol{\imath}{\mathord}{letters}{"7B}
\expandafter\let\csname jmath\endcsname\relax
\DeclareMathSymbol{\jmath}{\mathord}{letters}{"7C}
\expandafter\let\csname ell\endcsname\relax
\DeclareMathSymbol{\ell}{\mathord}{letters}{"60}
\expandafter\let\csname wp\endcsname\relax
\DeclareMathSymbol{\wp}{\mathord}{letters}{"7D}
\expandafter\let\csname Re\endcsname\relax
\DeclareMathSymbol{\Re}{\mathord}{symbols}{"3C}
\expandafter\let\csname Im\endcsname\relax
\DeclareMathSymbol{\Im}{\mathord}{symbols}{"3D}
\expandafter\let\csname partial\endcsname\relax
\DeclareMathSymbol{\partial}{\mathord}{letters}{"40}
\expandafter\let\csname infty\endcsname\relax
\DeclareMathSymbol{\infty}{\mathord}{symbols}{"31}
\expandafter\let\csname prime\endcsname\relax
\DeclareMathSymbol{\prime}{\mathord}{symbols}{"30}
\expandafter\let\csname emptyset\endcsname\relax
\DeclareMathSymbol{\emptyset}{\mathord}{symbols}{"3B}
\expandafter\let\csname nabla\endcsname\relax
\DeclareMathSymbol{\nabla}{\mathord}{symbols}{"72}
\expandafter\let\csname top\endcsname\relax
\DeclareMathSymbol{\top}{\mathord}{symbols}{"3E}
\expandafter\let\csname bot\endcsname\relax
\DeclareMathSymbol{\bot}{\mathord}{symbols}{"3F}
\expandafter\let\csname triangle\endcsname\relax
\DeclareMathSymbol{\triangle}{\mathord}{symbols}{"34}
\expandafter\let\csname forall\endcsname\relax
\DeclareMathSymbol{\forall}{\mathord}{symbols}{"38}
\expandafter\let\csname exists\endcsname\relax
\DeclareMathSymbol{\exists}{\mathord}{symbols}{"39}
\expandafter\let\csname neg\endcsname\relax
\DeclareMathSymbol{\neg}{\mathord}{symbols}{"3A}
\expandafter\let\csname lnot\endcsname\relax
\DeclareMathSymbol{\lnot}{\mathord}{symbols}{"3A}
\expandafter\let\csname flat\endcsname\relax
\DeclareMathSymbol{\flat}{\mathord}{letters}{"5B}
\expandafter\let\csname natural\endcsname\relax
\DeclareMathSymbol{\natural}{\mathord}{letters}{"5C}
\expandafter\let\csname sharp\endcsname\relax
\DeclareMathSymbol{\sharp}{\mathord}{letters}{"5D}
\expandafter\let\csname clubsuit\endcsname\relax
\DeclareMathSymbol{\clubsuit}{\mathord}{symbols}{"7C}
\expandafter\let\csname diamondsuit\endcsname\relax
\DeclareMathSymbol{\diamondsuit}{\mathord}{symbols}{"7D}
\expandafter\let\csname heartsuit\endcsname\relax
\DeclareMathSymbol{\heartsuit}{\mathord}{symbols}{"7E}
\expandafter\let\csname spadesuit\endcsname\relax
\DeclareMathSymbol{\spadesuit}{\mathord}{symbols}{"7F}
\expandafter\let\csname hbar\endcsname\relax
\DeclareRobustCommand\hbar{{\mathchar'26\mkern-9muh}}
\expandafter\let\csname surd\endcsname\relax

\expandafter\let\csname angle\endcsname\relax
\DeclareRobustCommand\angle{{\vbox{\ialign{$\m@th\scriptstyle##$\crcr
      \not\mathrel{\mkern14mu}\crcr
      \noalign{\nointerlineskip}
      \mkern2.5mu\leaders\hrule \@height.34pt\hfill\mkern2.5mu\crcr}}}}
\expandafter\let\csname coprod\endcsname\relax
\DeclareMathSymbol{\coprod}{\mathop}{largesymbols}{"60}
\expandafter\let\csname bigvee\endcsname\relax
\DeclareMathSymbol{\bigvee}{\mathop}{largesymbols}{"57}
\expandafter\let\csname bigwedge\endcsname\relax
\DeclareMathSymbol{\bigwedge}{\mathop}{largesymbols}{"56}
\expandafter\let\csname biguplus\endcsname\relax
\DeclareMathSymbol{\biguplus}{\mathop}{largesymbols}{"55}
\expandafter\let\csname bigcap\endcsname\relax
\DeclareMathSymbol{\bigcap}{\mathop}{largesymbols}{"54}
\expandafter\let\csname bigcup\endcsname\relax
\DeclareMathSymbol{\bigcup}{\mathop}{largesymbols}{"53}
\expandafter\let\csname intop\endcsname\relax
\DeclareMathSymbol{\intop}{\mathop}{largesymbols}{"52}
\expandafter\let\csname prod\endcsname\relax
\DeclareMathSymbol{\prod}{\mathop}{largesymbols}{"51}
\expandafter\let\csname sum\endcsname\relax
\DeclareMathSymbol{\sum}{\mathop}{largesymbols}{"50}
\expandafter\let\csname bigotimes\endcsname\relax
\DeclareMathSymbol{\bigotimes}{\mathop}{largesymbols}{"4E}
\expandafter\let\csname bigoplus\endcsname\relax
\DeclareMathSymbol{\bigoplus}{\mathop}{largesymbols}{"4C}
\expandafter\let\csname bigodot\endcsname\relax
\DeclareMathSymbol{\bigodot}{\mathop}{largesymbols}{"4A}
\expandafter\let\csname ointop\endcsname\relax
\DeclareMathSymbol{\ointop}{\mathop}{largesymbols}{"48}
\expandafter\let\csname bigsqcup\endcsname\relax
\DeclareMathSymbol{\bigsqcup}{\mathop}{largesymbols}{"46}
\expandafter\let\csname smallint\endcsname\relax
\DeclareMathSymbol{\smallint}{\mathop}{symbols}{"73}
\expandafter\let\csname triangleleft\endcsname\relax
\DeclareMathSymbol{\triangleleft}{\mathbin}{letters}{"2F}
\expandafter\let\csname triangleright\endcsname\relax
\DeclareMathSymbol{\triangleright}{\mathbin}{letters}{"2E}
\expandafter\let\csname bigtriangleup\endcsname\relax
\DeclareMathSymbol{\bigtriangleup}{\mathbin}{symbols}{"34}
\expandafter\let\csname bigtriangledown\endcsname\relax
\DeclareMathSymbol{\bigtriangledown}{\mathbin}{symbols}{"35}
\expandafter\let\csname varbigtriangleup\endcsname\relax
\DeclareMathSymbol{\varbigtriangleup}{\mathbin}{symbols}{"34}
\expandafter\let\csname varbigtriangledown\endcsname\relax
\DeclareMathSymbol{\varbigtriangledown}{\mathbin}{symbols}{"35}
\expandafter\let\csname wedge\endcsname\relax
\DeclareMathSymbol{\wedge}{\mathbin}{symbols}{"5E}
\expandafter\let\csname vee\endcsname\relax
\DeclareMathSymbol{\vee}{\mathbin}{symbols}{"5F}
\expandafter\let\csname land\endcsname\relax
\DeclareMathSymbol{\land}{\mathbin}{symbols}{"5E}
\expandafter\let\csname lor\endcsname\relax
\DeclareMathSymbol{\lor}{\mathbin}{symbols}{"5F}
\expandafter\let\csname cap\endcsname\relax
\DeclareMathSymbol{\cap}{\mathbin}{symbols}{"5C}
\expandafter\let\csname cup\endcsname\relax
\DeclareMathSymbol{\cup}{\mathbin}{symbols}{"5B}
\expandafter\let\csname ddagger\endcsname\relax
\DeclareMathSymbol{\ddagger}{\mathbin}{symbols}{"7A}
\expandafter\let\csname dagger\endcsname\relax
\DeclareMathSymbol{\dagger}{\mathbin}{symbols}{"79}
\expandafter\let\csname sqcap\endcsname\relax
\DeclareMathSymbol{\sqcap}{\mathbin}{symbols}{"75}
\expandafter\let\csname sqcup\endcsname\relax
\DeclareMathSymbol{\sqcup}{\mathbin}{symbols}{"74}
\expandafter\let\csname uplus\endcsname\relax
\DeclareMathSymbol{\uplus}{\mathbin}{symbols}{"5D}
\expandafter\let\csname amalg\endcsname\relax
\DeclareMathSymbol{\amalg}{\mathbin}{symbols}{"71}
\expandafter\let\csname diamond\endcsname\relax
\DeclareMathSymbol{\diamond}{\mathbin}{symbols}{"05}
\expandafter\let\csname bullet\endcsname\relax
\DeclareMathSymbol{\bullet}{\mathbin}{symbols}{"0F}
\expandafter\let\csname wr\endcsname\relax
\DeclareMathSymbol{\wr}{\mathbin}{symbols}{"6F}
\expandafter\let\csname div\endcsname\relax
\DeclareMathSymbol{\div}{\mathbin}{symbols}{"04}
\expandafter\let\csname odot\endcsname\relax
\DeclareMathSymbol{\odot}{\mathbin}{symbols}{"0C}
\expandafter\let\csname oslash\endcsname\relax
\DeclareMathSymbol{\oslash}{\mathbin}{symbols}{"0B}
\expandafter\let\csname otimes\endcsname\relax
\DeclareMathSymbol{\otimes}{\mathbin}{symbols}{"0A}
\expandafter\let\csname ominus\endcsname\relax
\DeclareMathSymbol{\ominus}{\mathbin}{symbols}{"09}
\expandafter\let\csname oplus\endcsname\relax
\DeclareMathSymbol{\oplus}{\mathbin}{symbols}{"08}
\expandafter\let\csname mp\endcsname\relax
\DeclareMathSymbol{\mp}{\mathbin}{symbols}{"07}
\expandafter\let\csname pm\endcsname\relax
\DeclareMathSymbol{\pm}{\mathbin}{symbols}{"06}
\expandafter\let\csname circ\endcsname\relax
\DeclareMathSymbol{\circ}{\mathbin}{symbols}{"0E}
\expandafter\let\csname bigcirc\endcsname\relax
\DeclareMathSymbol{\bigcirc}{\mathbin}{symbols}{"0D}
\expandafter\let\csname setminus\endcsname\relax
\DeclareMathSymbol{\setminus}{\mathbin}{symbols}{"6E}
\expandafter\let\csname cdot\endcsname\relax
\DeclareMathSymbol{\cdot}{\mathbin}{symbols}{"01}
\expandafter\let\csname ast\endcsname\relax
\DeclareMathSymbol{\ast}{\mathbin}{symbols}{"03}
\expandafter\let\csname times\endcsname\relax
\DeclareMathSymbol{\times}{\mathbin}{symbols}{"02}
\expandafter\let\csname star\endcsname\relax
\DeclareMathSymbol{\star}{\mathbin}{letters}{"3F}
\expandafter\let\csname propto\endcsname\relax
\DeclareMathSymbol{\propto}{\mathrel}{symbols}{"2F}
\expandafter\let\csname sqsubseteq\endcsname\relax
\DeclareMathSymbol{\sqsubseteq}{\mathrel}{symbols}{"76}
\expandafter\let\csname sqsupseteq\endcsname\relax
\DeclareMathSymbol{\sqsupseteq}{\mathrel}{symbols}{"77}
\expandafter\let\csname parallel\endcsname\relax
\DeclareMathSymbol{\parallel}{\mathrel}{symbols}{"6B}
\expandafter\let\csname mid\endcsname\relax
\DeclareMathSymbol{\mid}{\mathrel}{symbols}{"6A}
\expandafter\let\csname dashv\endcsname\relax
\DeclareMathSymbol{\dashv}{\mathrel}{symbols}{"61}
\expandafter\let\csname vdash\endcsname\relax
\DeclareMathSymbol{\vdash}{\mathrel}{symbols}{"60}
\expandafter\let\csname nearrow\endcsname\relax
\DeclareMathSymbol{\nearrow}{\mathrel}{symbols}{"25}
\expandafter\let\csname searrow\endcsname\relax
\DeclareMathSymbol{\searrow}{\mathrel}{symbols}{"26}
\expandafter\let\csname nwarrow\endcsname\relax
\DeclareMathSymbol{\nwarrow}{\mathrel}{symbols}{"2D}
\expandafter\let\csname swarrow\endcsname\relax
\DeclareMathSymbol{\swarrow}{\mathrel}{symbols}{"2E}
\expandafter\let\csname Leftrightarrow\endcsname\relax
\DeclareMathSymbol{\Leftrightarrow}{\mathrel}{symbols}{"2C}
\expandafter\let\csname Leftarrow\endcsname\relax
\DeclareMathSymbol{\Leftarrow}{\mathrel}{symbols}{"28}
\expandafter\let\csname Rightarrow\endcsname\relax
\DeclareMathSymbol{\Rightarrow}{\mathrel}{symbols}{"29}
\expandafter\let\csname neq\endcsname\relax

\expandafter\let\csname ne\endcsname\relax
\DeclareRobustCommand\ne{\not=}
\expandafter\let\csname leq\endcsname\relax
\DeclareMathSymbol{\leq}{\mathrel}{symbols}{"14}
\expandafter\let\csname geq\endcsname\relax
\DeclareMathSymbol{\geq}{\mathrel}{symbols}{"15}
\expandafter\let\csname le\endcsname\relax
\DeclareMathSymbol{\le}{\mathrel}{symbols}{"14}
\expandafter\let\csname ge\endcsname\relax
\DeclareMathSymbol{\ge}{\mathrel}{symbols}{"15}
\expandafter\let\csname succ\endcsname\relax
\DeclareMathSymbol{\succ}{\mathrel}{symbols}{"1F}
\expandafter\let\csname prec\endcsname\relax
\DeclareMathSymbol{\prec}{\mathrel}{symbols}{"1E}
\expandafter\let\csname approx\endcsname\relax
\DeclareMathSymbol{\approx}{\mathrel}{symbols}{"19}
\expandafter\let\csname succeq\endcsname\relax
\DeclareMathSymbol{\succeq}{\mathrel}{symbols}{"17}
\expandafter\let\csname preceq\endcsname\relax
\DeclareMathSymbol{\preceq}{\mathrel}{symbols}{"16}
\expandafter\let\csname supset\endcsname\relax
\DeclareMathSymbol{\supset}{\mathrel}{symbols}{"1B}
\expandafter\let\csname subset\endcsname\relax
\DeclareMathSymbol{\subset}{\mathrel}{symbols}{"1A}
\expandafter\let\csname supseteq\endcsname\relax
\DeclareMathSymbol{\supseteq}{\mathrel}{symbols}{"13}
\expandafter\let\csname subseteq\endcsname\relax
\DeclareMathSymbol{\subseteq}{\mathrel}{symbols}{"12}
\expandafter\let\csname in\endcsname\relax
\DeclareMathSymbol{\in}{\mathrel}{symbols}{"32}
\expandafter\let\csname ni\endcsname\relax
\DeclareMathSymbol{\ni}{\mathrel}{symbols}{"33}
\expandafter\let\csname owns\endcsname\relax
\DeclareMathSymbol{\owns}{\mathrel}{symbols}{"33}
\expandafter\let\csname gg\endcsname\relax
\DeclareMathSymbol{\gg}{\mathrel}{symbols}{"1D}
\expandafter\let\csname ll\endcsname\relax
\DeclareMathSymbol{\ll}{\mathrel}{symbols}{"1C}
\expandafter\let\csname not\endcsname\relax
\DeclareMathSymbol{\not}{\mathrel}{symbols}{"36}
\expandafter\let\csname leftrightarrow\endcsname\relax
\DeclareMathSymbol{\leftrightarrow}{\mathrel}{symbols}{"24}
\expandafter\let\csname leftarrow\endcsname\relax
\DeclareMathSymbol{\leftarrow}{\mathrel}{symbols}{"20}
\expandafter\let\csname rightarrow\endcsname\relax
\DeclareMathSymbol{\rightarrow}{\mathrel}{symbols}{"21}
\expandafter\let\csname gets\endcsname\relax
\DeclareMathSymbol{\gets}{\mathrel}{symbols}{"20}
\expandafter\let\csname to\endcsname\relax
\DeclareMathSymbol{\to}{\mathrel}{symbols}{"21}
\expandafter\let\csname mapstochar\endcsname\relax
\DeclareMathSymbol{\mapstochar}{\mathrel}{symbols}{"37}
\expandafter\let\csname sim\endcsname\relax
\DeclareMathSymbol{\sim}{\mathrel}{symbols}{"18}
\expandafter\let\csname simeq\endcsname\relax
\DeclareMathSymbol{\simeq}{\mathrel}{symbols}{"27}
\expandafter\let\csname perp\endcsname\relax
\DeclareMathSymbol{\perp}{\mathrel}{symbols}{"3F}
\expandafter\let\csname equiv\endcsname\relax
\DeclareMathSymbol{\equiv}{\mathrel}{symbols}{"11}
\expandafter\let\csname asymp\endcsname\relax
\DeclareMathSymbol{\asymp}{\mathrel}{symbols}{"10}
\expandafter\let\csname smile\endcsname\relax
\DeclareMathSymbol{\smile}{\mathrel}{letters}{"5E}
\expandafter\let\csname frown\endcsname\relax
\DeclareMathSymbol{\frown}{\mathrel}{letters}{"5F}
\expandafter\let\csname leftharpoonup\endcsname\relax
\DeclareMathSymbol{\leftharpoonup}{\mathrel}{letters}{"28}
\expandafter\let\csname leftharpoondown\endcsname\relax
\DeclareMathSymbol{\leftharpoondown}{\mathrel}{letters}{"29}
\expandafter\let\csname rightharpoonup\endcsname\relax
\DeclareMathSymbol{\rightharpoonup}{\mathrel}{letters}{"2A}
\expandafter\let\csname rightharpoondown\endcsname\relax
\DeclareMathSymbol{\rightharpoondown}{\mathrel}{letters}{"2B}
\expandafter\let\csname cong\endcsname\relax
\DeclareRobustCommand
  \cong{\mathrel{\mathpalette\@vereq\sim}}
\expandafter\let\csname notin\endcsname\relax
\DeclareRobustCommand
  \notin{\mathrel{\m@th\mathpalette\c@ncel\in}}
\expandafter\let\csname rightleftharpoons\endcsname\relax
\DeclareRobustCommand
  \rightleftharpoons{\mathrel{\mathpalette\rlh@{}}}
\expandafter\let\csname joinrel\endcsname\relax

\expandafter\let\csname relbar\endcsname\relax

\expandafter\let\csname Relbar\endcsname\relax

\expandafter\let\csname lhook\endcsname\relax
\DeclareMathSymbol{\lhook}{\mathrel}{letters}{"2C}
\expandafter\let\csname rhook\endcsname\relax
\DeclareMathSymbol{\rhook}{\mathrel}{letters}{"2D}
\expandafter\let\csname bowtie\endcsname\relax

\expandafter\let\csname models\endcsname\relax

\expandafter\let\csname ldotp\endcsname\relax
\DeclareMathSymbol{\ldotp}{\mathpunct}{letters}{"3A}
\expandafter\let\csname cdotp\endcsname\relax
\DeclareMathSymbol{\cdotp}{\mathpunct}{symbols}{"01}
\expandafter\let\csname colon\endcsname\relax
\DeclareMathSymbol{\colon}{\mathpunct}{operators}{"3A}
\expandafter\let\csname cdots\endcsname\relax
\DeclareRobustCommand
  \cdots{\mathinner{\cdotp\cdotp\cdotp}}
\expandafter\let\csname vdots\endcsname\relax
\DeclareRobustCommand
  \vdots{\vbox{\baselineskip4\p@ \lineskiplimit\z@
    \kern6\p@\hbox{.}\hbox{.}\hbox{.}}}
\expandafter\let\csname ddots\endcsname\relax
\DeclareRobustCommand
  \ddots{\mathinner{\mkern1mu\raise7\p@
    \vbox{\kern7\p@\hbox{.}}\mkern2mu
    \raise4\p@\hbox{.}\mkern2mu\raise\p@\hbox{.}\mkern1mu}}
\expandafter\let\csname acute\endcsname\relax
\DeclareMathAccent{\acute}{\mathalpha}{operators}{"13}
\expandafter\let\csname grave\endcsname\relax
\DeclareMathAccent{\grave}{\mathalpha}{operators}{"12}
\expandafter\let\csname ddot\endcsname\relax
\DeclareMathAccent{\ddot}{\mathalpha}{operators}{"7F}
\expandafter\let\csname tilde\endcsname\relax
\DeclareMathAccent{\tilde}{\mathalpha}{operators}{"7E}
\expandafter\let\csname bar\endcsname\relax
\DeclareMathAccent{\bar}{\mathalpha}{operators}{"16}
\expandafter\let\csname breve\endcsname\relax
\DeclareMathAccent{\breve}{\mathalpha}{operators}{"15}
\expandafter\let\csname check\endcsname\relax
\DeclareMathAccent{\check}{\mathalpha}{operators}{"14}
\expandafter\let\csname hat\endcsname\relax
\DeclareMathAccent{\hat}{\mathalpha}{operators}{"5E}
\expandafter\let\csname vec\endcsname\relax
\DeclareMathAccent{\vec}{\mathord}{letters}{"7E}
\expandafter\let\csname dot\endcsname\relax
\DeclareMathAccent{\dot}{\mathalpha}{operators}{"5F}
\expandafter\let\csname widetilde\endcsname\relax
\DeclareMathAccent{\widetilde}{\mathord}{largesymbols}{"65}
\expandafter\let\csname widehat\endcsname\relax
\DeclareMathAccent{\widehat}{\mathord}{largesymbols}{"62}
\expandafter\let\csname mathring\endcsname\relax
\DeclareMathAccent{\mathring}{\mathalpha}{operators}{"17}
\expandafter\let\csname sqrtsign\endcsname\relax
\DeclareMathRadical{\sqrtsign}{symbols}{"70}{largesymbols}{"70}
\expandafter\let\csname rightarrowfill\endcsname\relax
\DeclareRobustCommand\rightarrowfill{$\m@th\smash-\mkern-7mu
  \cleaders\hbox{$\mkern-2mu\smash-\mkern-2mu$}\hfill
  \mkern-7mu\mathord\rightarrow$}
\expandafter\let\csname leftarrowfill\endcsname\relax
\DeclareRobustCommand\leftarrowfill{$\m@th\mathord\leftarrow\mkern-7mu
  \cleaders\hbox{$\mkern-2mu\smash-\mkern-2mu$}\hfill
  \mkern-7mu\smash-$}
\expandafter\let\csname braceld\endcsname\relax
\DeclareMathSymbol{\braceld}{\mathord}{largesymbols}{"7A}
\expandafter\let\csname bracerd\endcsname\relax
\DeclareMathSymbol{\bracerd}{\mathord}{largesymbols}{"7B}
\expandafter\let\csname bracelu\endcsname\relax
\DeclareMathSymbol{\bracelu}{\mathord}{largesymbols}{"7C}
\expandafter\let\csname braceru\endcsname\relax
\DeclareMathSymbol{\braceru}{\mathord}{largesymbols}{"7D}
\expandafter\let\csname downbracefill\endcsname\relax
\DeclareRobustCommand\downbracefill{$\m@th \setbox\z@\hbox{$\braceld$}
  \braceld\leaders\vrule \@height\ht\z@ \@depth\z@\hfill\braceru
  \bracelu\leaders\vrule \@height\ht\z@ \@depth\z@\hfill\bracerd$}
\expandafter\let\csname upbracefill\endcsname\relax
\DeclareRobustCommand\upbracefill{$\m@th \setbox\z@\hbox{$\braceld$}
  \bracelu\leaders\vrule \@height\ht\z@ \@depth\z@\hfill\bracerd
  \braceld\leaders\vrule \@height\ht\z@ \@depth\z@\hfill\braceru$}
\expandafter\let\csname mathparagraph\endcsname\relax
\DeclareMathSymbol{\mathparagraph}{\mathord}{symbols}{"7B}
\expandafter\let\csname mathsection\endcsname\relax
\DeclareMathSymbol{\mathsection}{\mathord}{symbols}{"78}
\expandafter\let\csname mathdollar\endcsname\relax
\DeclareMathSymbol{\mathdollar}{\mathord}{operators}{"24}
\expandafter\let\csname mathsterling\endcsname\relax

\expandafter\let\csname mathunderscore\endcsname\relax
\DeclareRobustCommand\mathunderscore{\kern.06em\vbox{\hrule\@width.3em}}
\expandafter\let\csname mathellipsis\endcsname\relax

\expandafter\let\csname square\endcsname\relax
\DeclareMathSymbol{\square}{\mathord}{AMSa}{"03}
\expandafter\let\csname lozenge\endcsname\relax
\DeclareMathSymbol{\lozenge}{\mathord}{AMSa}{"06}
\expandafter\let\csname lhd\endcsname\relax
\DeclareMathSymbol{\lhd}{\mathbin}{AMSa}{"43}
\expandafter\let\csname unlhd\endcsname\relax
\DeclareMathSymbol{\unlhd}{\mathbin}{AMSa}{"45}
\expandafter\let\csname rhd\endcsname\relax
\DeclareMathSymbol{\rhd}{\mathbin}{AMSa}{"42}
\expandafter\let\csname unrhd\endcsname\relax
\DeclareMathSymbol{\unrhd}{\mathbin}{AMSa}{"44}
\expandafter\let\csname dabar@\endcsname\relax
\DeclareMathSymbol{\dabar@}{\mathord}{AMSa}{"39}
\expandafter\let\csname boxdot\endcsname\relax
\DeclareMathSymbol{\boxdot}{\mathbin}{AMSa}{"00}
\expandafter\let\csname boxplus\endcsname\relax
\DeclareMathSymbol{\boxplus}{\mathbin}{AMSa}{"01}
\expandafter\let\csname boxtimes\endcsname\relax
\DeclareMathSymbol{\boxtimes}{\mathbin}{AMSa}{"02}
\expandafter\let\csname square\endcsname\relax
\DeclareMathSymbol{\square}{\mathord}{AMSa}{"03}
\expandafter\let\csname blacksquare\endcsname\relax
\DeclareMathSymbol{\blacksquare}{\mathord}{AMSa}{"04}
\expandafter\let\csname centerdot\endcsname\relax
\DeclareMathSymbol{\centerdot}{\mathbin}{AMSa}{"05}
\expandafter\let\csname lozenge\endcsname\relax
\DeclareMathSymbol{\lozenge}{\mathord}{AMSa}{"06}
\expandafter\let\csname blacklozenge\endcsname\relax
\DeclareMathSymbol{\blacklozenge}{\mathord}{AMSa}{"07}
\expandafter\let\csname circlearrowright\endcsname\relax
\DeclareMathSymbol{\circlearrowright}{\mathrel}{AMSa}{"08}
\expandafter\let\csname circlearrowleft\endcsname\relax
\DeclareMathSymbol{\circlearrowleft}{\mathrel}{AMSa}{"09}
\expandafter\let\csname leftrightharpoons\endcsname\relax
\DeclareMathSymbol{\leftrightharpoons}{\mathrel}{AMSa}{"0B}
\expandafter\let\csname boxminus\endcsname\relax
\DeclareMathSymbol{\boxminus}{\mathbin}{AMSa}{"0C}
\expandafter\let\csname Vdash\endcsname\relax
\DeclareMathSymbol{\Vdash}{\mathrel}{AMSa}{"0D}
\expandafter\let\csname Vvdash\endcsname\relax
\DeclareMathSymbol{\Vvdash}{\mathrel}{AMSa}{"0E}
\expandafter\let\csname vDash\endcsname\relax
\DeclareMathSymbol{\vDash}{\mathrel}{AMSa}{"0F}
\expandafter\let\csname twoheadrightarrow\endcsname\relax
\DeclareMathSymbol{\twoheadrightarrow}{\mathrel}{AMSa}{"10}
\expandafter\let\csname twoheadleftarrow\endcsname\relax
\DeclareMathSymbol{\twoheadleftarrow}{\mathrel}{AMSa}{"11}
\expandafter\let\csname leftleftarrows\endcsname\relax
\DeclareMathSymbol{\leftleftarrows}{\mathrel}{AMSa}{"12}
\expandafter\let\csname rightrightarrows\endcsname\relax
\DeclareMathSymbol{\rightrightarrows}{\mathrel}{AMSa}{"13}
\expandafter\let\csname upuparrows\endcsname\relax
\DeclareMathSymbol{\upuparrows}{\mathrel}{AMSa}{"14}
\expandafter\let\csname downdownarrows\endcsname\relax
\DeclareMathSymbol{\downdownarrows}{\mathrel}{AMSa}{"15}
\expandafter\let\csname upharpoonright\endcsname\relax
\DeclareMathSymbol{\upharpoonright}{\mathrel}{AMSa}{"16}
\expandafter\let\csname downharpoonright\endcsname\relax
\DeclareMathSymbol{\downharpoonright}{\mathrel}{AMSa}{"17}
\expandafter\let\csname upharpoonleft\endcsname\relax
\DeclareMathSymbol{\upharpoonleft}{\mathrel}{AMSa}{"18}
\expandafter\let\csname rightarrowtail\endcsname\relax
\DeclareMathSymbol{\rightarrowtail}{\mathrel}{AMSa}{"1A}
\expandafter\let\csname leftarrowtail\endcsname\relax
\DeclareMathSymbol{\leftarrowtail}{\mathrel}{AMSa}{"1B}
\expandafter\let\csname Lsh\endcsname\relax
\DeclareMathSymbol{\Lsh}{\mathrel}{AMSa}{"1E}
\expandafter\let\csname Rsh\endcsname\relax
\DeclareMathSymbol{\Rsh}{\mathrel}{AMSa}{"1F}
\expandafter\let\csname rightsquigarrow\endcsname\relax
\DeclareMathSymbol{\rightsquigarrow}{\mathrel}{AMSa}{"20}
\expandafter\let\csname looparrowleft\endcsname\relax
\DeclareMathSymbol{\looparrowleft}{\mathrel}{AMSa}{"22}
\expandafter\let\csname looparrowright\endcsname\relax
\DeclareMathSymbol{\looparrowright}{\mathrel}{AMSa}{"23}
\expandafter\let\csname circeq\endcsname\relax
\DeclareMathSymbol{\circeq}{\mathrel}{AMSa}{"24}
\expandafter\let\csname succsim\endcsname\relax
\DeclareMathSymbol{\succsim}{\mathrel}{AMSa}{"25}
\expandafter\let\csname gtrsim\endcsname\relax
\DeclareMathSymbol{\gtrsim}{\mathrel}{AMSa}{"26}
\expandafter\let\csname gtrapprox\endcsname\relax
\DeclareMathSymbol{\gtrapprox}{\mathrel}{AMSa}{"27}
\expandafter\let\csname multimap\endcsname\relax
\DeclareMathSymbol{\multimap}{\mathrel}{AMSa}{"28}
\expandafter\let\csname therefore\endcsname\relax
\DeclareMathSymbol{\therefore}{\mathrel}{AMSa}{"29}
\expandafter\let\csname because\endcsname\relax
\DeclareMathSymbol{\because}{\mathrel}{AMSa}{"2A}
\expandafter\let\csname doteqdot\endcsname\relax
\DeclareMathSymbol{\doteqdot}{\mathrel}{AMSa}{"2B}
\expandafter\let\csname triangleq\endcsname\relax
\DeclareMathSymbol{\triangleq}{\mathrel}{AMSa}{"2C}
\expandafter\let\csname precsim\endcsname\relax
\DeclareMathSymbol{\precsim}{\mathrel}{AMSa}{"2D}
\expandafter\let\csname lesssim\endcsname\relax
\DeclareMathSymbol{\lesssim}{\mathrel}{AMSa}{"2E}
\expandafter\let\csname lessapprox\endcsname\relax
\DeclareMathSymbol{\lessapprox}{\mathrel}{AMSa}{"2F}
\expandafter\let\csname eqslantless\endcsname\relax
\DeclareMathSymbol{\eqslantless}{\mathrel}{AMSa}{"30}
\expandafter\let\csname eqslantgtr\endcsname\relax
\DeclareMathSymbol{\eqslantgtr}{\mathrel}{AMSa}{"31}
\expandafter\let\csname curlyeqprec\endcsname\relax
\DeclareMathSymbol{\curlyeqprec}{\mathrel}{AMSa}{"32}
\expandafter\let\csname curlyeqsucc\endcsname\relax
\DeclareMathSymbol{\curlyeqsucc}{\mathrel}{AMSa}{"33}
\expandafter\let\csname preccurlyeq\endcsname\relax
\DeclareMathSymbol{\preccurlyeq}{\mathrel}{AMSa}{"34}
\expandafter\let\csname leqq\endcsname\relax
\DeclareMathSymbol{\leqq}{\mathrel}{AMSa}{"35}
\expandafter\let\csname leqslant\endcsname\relax
\DeclareMathSymbol{\leqslant}{\mathrel}{AMSa}{"36}
\expandafter\let\csname lessgtr\endcsname\relax
\DeclareMathSymbol{\lessgtr}{\mathrel}{AMSa}{"37}
\expandafter\let\csname backprime\endcsname\relax
\DeclareMathSymbol{\backprime}{\mathord}{AMSa}{"38}
\expandafter\let\csname risingdotseq\endcsname\relax
\DeclareMathSymbol{\risingdotseq}{\mathrel}{AMSa}{"3A}
\expandafter\let\csname succcurlyeq\endcsname\relax
\DeclareMathSymbol{\succcurlyeq}{\mathrel}{AMSa}{"3C}
\expandafter\let\csname geqq\endcsname\relax
\DeclareMathSymbol{\geqq}{\mathrel}{AMSa}{"3D}
\expandafter\let\csname geqslant\endcsname\relax
\DeclareMathSymbol{\geqslant}{\mathrel}{AMSa}{"3E}
\expandafter\let\csname gtrless\endcsname\relax
\DeclareMathSymbol{\gtrless}{\mathrel}{AMSa}{"3F}
\expandafter\let\csname vartriangleleft\endcsname\relax
\DeclareMathSymbol{\vartriangleleft}{\mathrel}{AMSa}{"43}
\expandafter\let\csname trianglerighteq\endcsname\relax
\DeclareMathSymbol{\trianglerighteq}{\mathrel}{AMSa}{"44}
\expandafter\let\csname trianglelefteq\endcsname\relax
\DeclareMathSymbol{\trianglelefteq}{\mathrel}{AMSa}{"45}
\expandafter\let\csname bigstar\endcsname\relax
\DeclareMathSymbol{\bigstar}{\mathord}{AMSa}{"46}
\expandafter\let\csname between\endcsname\relax
\DeclareMathSymbol{\between}{\mathrel}{AMSa}{"47}
\expandafter\let\csname blacktriangledown\endcsname\relax
\DeclareMathSymbol{\blacktriangledown}{\mathord}{AMSa}{"48}
\expandafter\let\csname blacktriangleright\endcsname\relax
\DeclareMathSymbol{\blacktriangleright}{\mathrel}{AMSa}{"49}
\expandafter\let\csname blacktriangleleft\endcsname\relax
\DeclareMathSymbol{\blacktriangleleft}{\mathrel}{AMSa}{"4A}
\expandafter\let\csname vartriangle\endcsname\relax
\DeclareMathSymbol{\vartriangle}{\mathrel}{AMSa}{"4D}
\expandafter\let\csname blacktriangle\endcsname\relax
\DeclareMathSymbol{\blacktriangle}{\mathord}{AMSa}{"4E}
\expandafter\let\csname triangledown\endcsname\relax
\DeclareMathSymbol{\triangledown}{\mathord}{AMSa}{"4F}
\expandafter\let\csname eqcirc\endcsname\relax
\DeclareMathSymbol{\eqcirc}{\mathrel}{AMSa}{"50}
\expandafter\let\csname lesseqgtr\endcsname\relax
\DeclareMathSymbol{\lesseqgtr}{\mathrel}{AMSa}{"51}
\expandafter\let\csname gtreqless\endcsname\relax
\DeclareMathSymbol{\gtreqless}{\mathrel}{AMSa}{"52}
\expandafter\let\csname lesseqqgtr\endcsname\relax
\DeclareMathSymbol{\lesseqqgtr}{\mathrel}{AMSa}{"53}
\expandafter\let\csname gtreqqless\endcsname\relax
\DeclareMathSymbol{\gtreqqless}{\mathrel}{AMSa}{"54}
\expandafter\let\csname Rrightarrow\endcsname\relax
\DeclareMathSymbol{\Rrightarrow}{\mathrel}{AMSa}{"56}
\expandafter\let\csname Lleftarrow\endcsname\relax
\DeclareMathSymbol{\Lleftarrow}{\mathrel}{AMSa}{"57}
\expandafter\let\csname veebar\endcsname\relax
\DeclareMathSymbol{\veebar}{\mathbin}{AMSa}{"59}
\expandafter\let\csname barwedge\endcsname\relax
\DeclareMathSymbol{\barwedge}{\mathbin}{AMSa}{"5A}
\expandafter\let\csname doublebarwedge\endcsname\relax
\DeclareMathSymbol{\doublebarwedge}{\mathbin}{AMSa}{"5B}
\expandafter\let\csname measuredangle\endcsname\relax
\DeclareMathSymbol{\measuredangle}{\mathord}{AMSa}{"5D}
\expandafter\let\csname sphericalangle\endcsname\relax
\DeclareMathSymbol{\sphericalangle}{\mathord}{AMSa}{"5E}
\expandafter\let\csname varpropto\endcsname\relax
\DeclareMathSymbol{\varpropto}{\mathrel}{AMSa}{"5F}
\expandafter\let\csname smallsmile\endcsname\relax
\DeclareMathSymbol{\smallsmile}{\mathrel}{AMSa}{"60}
\expandafter\let\csname smallfrown\endcsname\relax
\DeclareMathSymbol{\smallfrown}{\mathrel}{AMSa}{"61}
\expandafter\let\csname Subset\endcsname\relax
\DeclareMathSymbol{\Subset}{\mathrel}{AMSa}{"62}
\expandafter\let\csname Supset\endcsname\relax
\DeclareMathSymbol{\Supset}{\mathrel}{AMSa}{"63}
\expandafter\let\csname Cup\endcsname\relax
\DeclareMathSymbol{\Cup}{\mathbin}{AMSa}{"64}
\expandafter\let\csname Cap\endcsname\relax
\DeclareMathSymbol{\Cap}{\mathbin}{AMSa}{"65}
\expandafter\let\csname curlywedge\endcsname\relax
\DeclareMathSymbol{\curlywedge}{\mathbin}{AMSa}{"66}
\expandafter\let\csname curlyvee\endcsname\relax
\DeclareMathSymbol{\curlyvee}{\mathbin}{AMSa}{"67}
\expandafter\let\csname leftthreetimes\endcsname\relax
\DeclareMathSymbol{\leftthreetimes}{\mathbin}{AMSa}{"68}
\expandafter\let\csname subseteqq\endcsname\relax
\DeclareMathSymbol{\subseteqq}{\mathrel}{AMSa}{"6A}
\expandafter\let\csname supseteqq\endcsname\relax
\DeclareMathSymbol{\supseteqq}{\mathrel}{AMSa}{"6B}
\expandafter\let\csname bumpeq\endcsname\relax
\DeclareMathSymbol{\bumpeq}{\mathrel}{AMSa}{"6C}
\expandafter\let\csname Bumpeq\endcsname\relax
\DeclareMathSymbol{\Bumpeq}{\mathrel}{AMSa}{"6D}
\expandafter\let\csname lll\endcsname\relax
\DeclareMathSymbol{\lll}{\mathrel}{AMSa}{"6E}
\expandafter\let\csname ggg\endcsname\relax
\DeclareMathSymbol{\ggg}{\mathrel}{AMSa}{"6F}
\expandafter\let\csname circledS\endcsname\relax
\DeclareMathSymbol{\circledS}{\mathord}{AMSa}{"73}
\expandafter\let\csname pitchfork\endcsname\relax
\DeclareMathSymbol{\pitchfork}{\mathrel}{AMSa}{"74}
\expandafter\let\csname dotplus\endcsname\relax
\DeclareMathSymbol{\dotplus}{\mathbin}{AMSa}{"75}
\expandafter\let\csname backsim\endcsname\relax
\DeclareMathSymbol{\backsim}{\mathrel}{AMSa}{"76}
\expandafter\let\csname backsimeq\endcsname\relax
\DeclareMathSymbol{\backsimeq}{\mathrel}{AMSa}{"77}
\expandafter\let\csname complement\endcsname\relax
\DeclareMathSymbol{\complement}{\mathord}{AMSa}{"7B}
\expandafter\let\csname intercal\endcsname\relax
\DeclareMathSymbol{\intercal}{\mathbin}{AMSa}{"7C}
\expandafter\let\csname circledcirc\endcsname\relax
\DeclareMathSymbol{\circledcirc}{\mathbin}{AMSa}{"7D}
\expandafter\let\csname circledast\endcsname\relax
\DeclareMathSymbol{\circledast}{\mathbin}{AMSa}{"7E}
\expandafter\let\csname circleddash\endcsname\relax
\DeclareMathSymbol{\circleddash}{\mathbin}{AMSa}{"7F}
\expandafter\let\csname lvertneqq\endcsname\relax
\DeclareMathSymbol{\lvertneqq}{\mathrel}{AMSb}{"00}
\expandafter\let\csname gvertneqq\endcsname\relax
\DeclareMathSymbol{\gvertneqq}{\mathrel}{AMSb}{"01}
\expandafter\let\csname nleq\endcsname\relax
\DeclareMathSymbol{\nleq}{\mathrel}{AMSb}{"02}
\expandafter\let\csname ngeq\endcsname\relax
\DeclareMathSymbol{\ngeq}{\mathrel}{AMSb}{"03}
\expandafter\let\csname nless\endcsname\relax
\DeclareMathSymbol{\nless}{\mathrel}{AMSb}{"04}
\expandafter\let\csname ngtr\endcsname\relax
\DeclareMathSymbol{\ngtr}{\mathrel}{AMSb}{"05}
\expandafter\let\csname nprec\endcsname\relax
\DeclareMathSymbol{\nprec}{\mathrel}{AMSb}{"06}
\expandafter\let\csname nsucc\endcsname\relax
\DeclareMathSymbol{\nsucc}{\mathrel}{AMSb}{"07}
\expandafter\let\csname lneqq\endcsname\relax
\DeclareMathSymbol{\lneqq}{\mathrel}{AMSb}{"08}
\expandafter\let\csname gneqq\endcsname\relax
\DeclareMathSymbol{\gneqq}{\mathrel}{AMSb}{"09}
\expandafter\let\csname nleqslant\endcsname\relax
\DeclareMathSymbol{\nleqslant}{\mathrel}{AMSb}{"0A}
\expandafter\let\csname ngeqslant\endcsname\relax
\DeclareMathSymbol{\ngeqslant}{\mathrel}{AMSb}{"0B}
\expandafter\let\csname lneq\endcsname\relax
\DeclareMathSymbol{\lneq}{\mathrel}{AMSb}{"0C}
\expandafter\let\csname gneq\endcsname\relax
\DeclareMathSymbol{\gneq}{\mathrel}{AMSb}{"0D}
\expandafter\let\csname npreceq\endcsname\relax
\DeclareMathSymbol{\npreceq}{\mathrel}{AMSb}{"0E}
\expandafter\let\csname nsucceq\endcsname\relax
\DeclareMathSymbol{\nsucceq}{\mathrel}{AMSb}{"0F}
\expandafter\let\csname precnsim\endcsname\relax
\DeclareMathSymbol{\precnsim}{\mathrel}{AMSb}{"10}
\expandafter\let\csname succnsim\endcsname\relax
\DeclareMathSymbol{\succnsim}{\mathrel}{AMSb}{"11}
\expandafter\let\csname lnsim\endcsname\relax
\DeclareMathSymbol{\lnsim}{\mathrel}{AMSb}{"12}
\expandafter\let\csname gnsim\endcsname\relax
\DeclareMathSymbol{\gnsim}{\mathrel}{AMSb}{"13}
\expandafter\let\csname nleqq\endcsname\relax
\DeclareMathSymbol{\nleqq}{\mathrel}{AMSb}{"14}
\expandafter\let\csname ngeqq\endcsname\relax
\DeclareMathSymbol{\ngeqq}{\mathrel}{AMSb}{"15}
\expandafter\let\csname precneqq\endcsname\relax
\DeclareMathSymbol{\precneqq}{\mathrel}{AMSb}{"16}
\expandafter\let\csname succneqq\endcsname\relax
\DeclareMathSymbol{\succneqq}{\mathrel}{AMSb}{"17}
\expandafter\let\csname precnapprox\endcsname\relax
\DeclareMathSymbol{\precnapprox}{\mathrel}{AMSb}{"18}
\expandafter\let\csname succnapprox\endcsname\relax
\DeclareMathSymbol{\succnapprox}{\mathrel}{AMSb}{"19}
\expandafter\let\csname lnapprox\endcsname\relax
\DeclareMathSymbol{\lnapprox}{\mathrel}{AMSb}{"1A}
\expandafter\let\csname gnapprox\endcsname\relax
\DeclareMathSymbol{\gnapprox}{\mathrel}{AMSb}{"1B}
\expandafter\let\csname nsim\endcsname\relax
\DeclareMathSymbol{\nsim}{\mathrel}{AMSb}{"1C}
\expandafter\let\csname ncong\endcsname\relax
\DeclareMathSymbol{\ncong}{\mathrel}{AMSb}{"1D}
\expandafter\let\csname diagup\endcsname\relax
\DeclareMathSymbol{\diagup}{\mathord}{AMSb}{"1E}
\expandafter\let\csname diagdown\endcsname\relax
\DeclareMathSymbol{\diagdown}{\mathord}{AMSb}{"1F}
\expandafter\let\csname varsubsetneq\endcsname\relax
\DeclareMathSymbol{\varsubsetneq}{\mathrel}{AMSb}{"20}
\expandafter\let\csname varsupsetneq\endcsname\relax
\DeclareMathSymbol{\varsupsetneq}{\mathrel}{AMSb}{"21}
\expandafter\let\csname nsubseteqq\endcsname\relax
\DeclareMathSymbol{\nsubseteqq}{\mathrel}{AMSb}{"22}
\expandafter\let\csname nsupseteqq\endcsname\relax
\DeclareMathSymbol{\nsupseteqq}{\mathrel}{AMSb}{"23}
\expandafter\let\csname subsetneqq\endcsname\relax
\DeclareMathSymbol{\subsetneqq}{\mathrel}{AMSb}{"24}
\expandafter\let\csname supsetneqq\endcsname\relax
\DeclareMathSymbol{\supsetneqq}{\mathrel}{AMSb}{"25}
\expandafter\let\csname varsubsetneqq\endcsname\relax
\DeclareMathSymbol{\varsubsetneqq}{\mathrel}{AMSb}{"26}
\expandafter\let\csname varsupsetneqq\endcsname\relax
\DeclareMathSymbol{\varsupsetneqq}{\mathrel}{AMSb}{"27}
\expandafter\let\csname subsetneq\endcsname\relax
\DeclareMathSymbol{\subsetneq}{\mathrel}{AMSb}{"28}
\expandafter\let\csname supsetneq\endcsname\relax
\DeclareMathSymbol{\supsetneq}{\mathrel}{AMSb}{"29}
\expandafter\let\csname nsubseteq\endcsname\relax
\DeclareMathSymbol{\nsubseteq}{\mathrel}{AMSb}{"2A}
\expandafter\let\csname nsupseteq\endcsname\relax
\DeclareMathSymbol{\nsupseteq}{\mathrel}{AMSb}{"2B}
\expandafter\let\csname nparallel\endcsname\relax
\DeclareMathSymbol{\nparallel}{\mathrel}{AMSb}{"2C}
\expandafter\let\csname nmid\endcsname\relax
\DeclareMathSymbol{\nmid}{\mathrel}{AMSb}{"2D}
\expandafter\let\csname nshortmid\endcsname\relax
\DeclareMathSymbol{\nshortmid}{\mathrel}{AMSb}{"2E}
\expandafter\let\csname nshortparallel\endcsname\relax
\DeclareMathSymbol{\nshortparallel}{\mathrel}{AMSb}{"2F}
\expandafter\let\csname nvdash\endcsname\relax
\DeclareMathSymbol{\nvdash}{\mathrel}{AMSb}{"30}
\expandafter\let\csname nVdash\endcsname\relax
\DeclareMathSymbol{\nVdash}{\mathrel}{AMSb}{"31}
\expandafter\let\csname nvDash\endcsname\relax
\DeclareMathSymbol{\nvDash}{\mathrel}{AMSb}{"32}
\expandafter\let\csname nVDash\endcsname\relax
\DeclareMathSymbol{\nVDash}{\mathrel}{AMSb}{"33}
\expandafter\let\csname ntriangleleft\endcsname\relax
\DeclareMathSymbol{\ntriangleleft}{\mathrel}{AMSb}{"36}
\expandafter\let\csname ntriangleright\endcsname\relax
\DeclareMathSymbol{\ntriangleright}{\mathrel}{AMSb}{"37}
\expandafter\let\csname nleftarrow\endcsname\relax
\DeclareMathSymbol{\nleftarrow}{\mathrel}{AMSb}{"38}
\expandafter\let\csname nrightarrow\endcsname\relax
\DeclareMathSymbol{\nrightarrow}{\mathrel}{AMSb}{"39}
\expandafter\let\csname nLeftarrow\endcsname\relax
\DeclareMathSymbol{\nLeftarrow}{\mathrel}{AMSb}{"3A}
\expandafter\let\csname nRightarrow\endcsname\relax
\DeclareMathSymbol{\nRightarrow}{\mathrel}{AMSb}{"3B}
\expandafter\let\csname divideontimes\endcsname\relax
\DeclareMathSymbol{\divideontimes}{\mathbin}{AMSb}{"3E}
\expandafter\let\csname varnothing\endcsname\relax
\DeclareMathSymbol{\varnothing}{\mathord}{AMSb}{"3F}
\expandafter\let\csname nexists\endcsname\relax
\DeclareMathSymbol{\nexists}{\mathord}{AMSb}{"40}
\expandafter\let\csname Finv\endcsname\relax
\DeclareMathSymbol{\Finv}{\mathord}{AMSb}{"60}
\expandafter\let\csname Game\endcsname\relax
\DeclareMathSymbol{\Game}{\mathord}{AMSb}{"61}
\expandafter\let\csname eth\endcsname\relax
\DeclareMathSymbol{\eth}{\mathord}{AMSb}{"67}
\expandafter\let\csname eqsim\endcsname\relax
\DeclareMathSymbol{\eqsim}{\mathrel}{AMSb}{"68}
\expandafter\let\csname beth\endcsname\relax
\DeclareMathSymbol{\beth}{\mathord}{AMSb}{"69}
\expandafter\let\csname gimel\endcsname\relax
\DeclareMathSymbol{\gimel}{\mathord}{AMSb}{"6A}
\expandafter\let\csname daleth\endcsname\relax
\DeclareMathSymbol{\daleth}{\mathord}{AMSb}{"6B}
\expandafter\let\csname lessdot\endcsname\relax
\DeclareMathSymbol{\lessdot}{\mathbin}{AMSb}{"6C}
\expandafter\let\csname gtrdot\endcsname\relax
\DeclareMathSymbol{\gtrdot}{\mathbin}{AMSb}{"6D}
\expandafter\let\csname ltimes\endcsname\relax
\DeclareMathSymbol{\ltimes}{\mathbin}{AMSb}{"6E}
\expandafter\let\csname rtimes\endcsname\relax
\DeclareMathSymbol{\rtimes}{\mathbin}{AMSb}{"6F}
\expandafter\let\csname shortmid\endcsname\relax
\DeclareMathSymbol{\shortmid}{\mathrel}{AMSb}{"70}
\expandafter\let\csname shortparallel\endcsname\relax
\DeclareMathSymbol{\shortparallel}{\mathrel}{AMSb}{"71}
\expandafter\let\csname smallsetminus\endcsname\relax
\DeclareMathSymbol{\smallsetminus}{\mathbin}{AMSb}{"72}
\expandafter\let\csname thicksim\endcsname\relax
\DeclareMathSymbol{\thicksim}{\mathrel}{AMSb}{"73}
\expandafter\let\csname thickapprox\endcsname\relax
\DeclareMathSymbol{\thickapprox}{\mathrel}{AMSb}{"74}
\expandafter\let\csname approxeq\endcsname\relax
\DeclareMathSymbol{\approxeq}{\mathrel}{AMSb}{"75}
\expandafter\let\csname succapprox\endcsname\relax
\DeclareMathSymbol{\succapprox}{\mathrel}{AMSb}{"76}
\expandafter\let\csname precapprox\endcsname\relax
\DeclareMathSymbol{\precapprox}{\mathrel}{AMSb}{"77}
\expandafter\let\csname curvearrowleft\endcsname\relax
\DeclareMathSymbol{\curvearrowleft}{\mathrel}{AMSb}{"78}
\expandafter\let\csname digamma\endcsname\relax
\DeclareMathSymbol{\digamma}{\mathord}{AMSb}{"7A}
\expandafter\let\csname varkappa\endcsname\relax
\DeclareMathSymbol{\varkappa}{\mathord}{AMSb}{"7B}
\expandafter\let\csname Bbbk\endcsname\relax
\DeclareMathSymbol{\Bbbk}{\mathord}{AMSb}{"7C}
\expandafter\let\csname hslash\endcsname\relax
\DeclareMathSymbol{\hslash}{\mathord}{AMSb}{"7D}
\expandafter\let\csname backepsilon\endcsname\relax
\DeclareMathSymbol{\backepsilon}{\mathrel}{AMSb}{"7F}
\expandafter\let\csname downharpoonleft\endcsname\relax
\DeclareMathSymbol{\downharpoonleft}{\mathrel}{AMSa}{"19}
\expandafter\let\csname leftrightarrows\endcsname\relax
\DeclareMathSymbol{\leftrightarrows}{\mathrel}{AMSa}{"1C}
\expandafter\let\csname rightleftarrows\endcsname\relax
\DeclareMathSymbol{\rightleftarrows}{\mathrel}{AMSa}{"1D}
\expandafter\let\csname leftrightsquigarrow\endcsname\relax
\DeclareMathSymbol{\leftrightsquigarrow}{\mathrel}{AMSa}{"21}
\expandafter\let\csname fallingdotseq\endcsname\relax
\DeclareMathSymbol{\fallingdotseq}{\mathrel}{AMSa}{"3B}
\expandafter\let\csname vartriangleright\endcsname\relax
\DeclareMathSymbol{\vartriangleright}{\mathrel}{AMSa}{"42}
\expandafter\let\csname rightthreetimes\endcsname\relax
\DeclareMathSymbol{\rightthreetimes}{\mathbin}{AMSa}{"69}
\expandafter\let\csname ntrianglerighteq\endcsname\relax
\DeclareMathSymbol{\ntrianglerighteq}{\mathrel}{AMSb}{"34}
\expandafter\let\csname ntrianglelefteq\endcsname\relax
\DeclareMathSymbol{\ntrianglelefteq}{\mathrel}{AMSb}{"35}
\expandafter\let\csname nLeftrightarrow\endcsname\relax
\DeclareMathSymbol{\nLeftrightarrow}{\mathrel}{AMSb}{"3C}
\expandafter\let\csname nleftrightarrow\endcsname\relax
\DeclareMathSymbol{\nleftrightarrow}{\mathrel}{AMSb}{"3D}
\expandafter\let\csname curvearrowright\endcsname\relax
\DeclareMathSymbol{\curvearrowright}{\mathrel}{AMSb}{"79}

\DeclareMathSizes{5}{5}{5}{5}
\DeclareMathSizes{6}{6}{5}{5}
\DeclareMathSizes{7}{7}{5}{5}
\DeclareMathSizes{8}{8}{6}{5}
\DeclareMathSizes{9}{9}{6}{5}
\DeclareMathSizes{\@xpt}{\@xpt}{7}{5}
\DeclareMathSizes{\@xipt}{\@xipt}{8}{6}
\DeclareMathSizes{\@xiipt}{\@xiipt}{8}{6}
\DeclareMathSizes{\@xivpt}{\@xivpt}{\@xpt}{7}
\DeclareMathSizes{\@xviipt}{\@xviipt}{\@xiipt}{\@xpt}
\DeclareMathSizes{\@xxpt}{\@xxpt}{\@xivpt}{\@xiipt}
\DeclareMathSizes{\@xxvpt}{\@xxvpt}{\@xxpt}{\@xviipt}

\mathcode`\%="0025
\mathcode`\&="0026
\mathcode`\#="0023
\mathcode`\\="026E

%% file: figures/xir-icons.tex
\definecolor{xirAppColor}{RGB}{70,76,86}
\definecolor{xirGatewayColor}{RGB}{32,99,155}
\definecolor{xirProtocolColor}{RGB}{194,103,22}
\definecolor{xirControlColor}{RGB}{42,126,92}
\definecolor{xirActorColor}{RGB}{111,72,146}
\definecolor{xirResultColor}{RGB}{180,86,20}

\tikzset{
  xir icon/.style={line width=0.62pt,line cap=round,line join=round},
  pics/xir-app/.style={code={
    \draw[xir icon,draw=xirAppColor,fill=xirAppColor!7] (-.42,-.28) rectangle (.42,.28);
    \draw[xir icon,draw=xirAppColor] (-.42,.13)--(.42,.13);
    \fill[xirAppColor] (-.31,.205) circle (.025) (-.23,.205) circle (.025);
    \draw[xir icon,draw=xirAppColor] (-.18,-.01)--(-.28,-.08)--(-.18,-.15)
      (.18,-.01)--(.28,-.08)--(.18,-.15) (-.05,-.17)--(.05,.01);
  }},
  pics/xir-gateway/.style={code={
    \path[draw=xirGatewayColor,fill=xirGatewayColor!9,xir icon]
      (-.34,.28)--(.34,.28)--(.31,-.04)
      .. controls (.28,-.22) and (.12,-.32) .. (0,-.37)
      .. controls (-.12,-.32) and (-.28,-.22) .. (-.31,-.04)--cycle;
    \draw[xir icon,draw=xirGatewayColor,-{Latex[length=1.2mm]}] (-.24,.03)--(.22,.03);
    \draw[xir icon,draw=xirGatewayColor,-{Latex[length=1.2mm]}] (.20,-.12)--(-.22,-.12);
  }},
  pics/xir-protocol/.style={code={
    \draw[xir icon,draw=xirProtocolColor,fill=xirProtocolColor!9,rounded corners=.03cm]
      (-.43,-.25) rectangle (.43,.25);
    \draw[xir icon,draw=xirProtocolColor] (-.41,.22)--(0,-.07)--(.41,.22);
    \draw[xir icon,draw=xirProtocolColor] (-.41,-.22)--(-.11,.02) (.41,-.22)--(.11,.02);
  }},
  pics/xir-registry/.style={code={
    \path[draw=xirControlColor,fill=xirControlColor!8,xir icon]
      (-.36,.19) arc[start angle=180,end angle=360,x radius=.36,y radius=.13]
      -- (.36,-.19) arc[start angle=0,end angle=-180,x radius=.36,y radius=.13]
      -- cycle;
    \draw[xir icon,draw=xirControlColor] (-.36,.19) arc[start angle=180,end angle=540,x radius=.36,y radius=.13];
    \draw[xir icon,draw=xirControlColor] (-.36,.02) arc[start angle=180,end angle=360,x radius=.36,y radius=.13];
  }},
  pics/xir-router/.style={code={
    \draw[xir icon,draw=xirActorColor,-{Latex[length=1.3mm]}] (-.40,0)--(-.08,0)--(.28,.25);
    \draw[xir icon,draw=xirActorColor,-{Latex[length=1.3mm]}] (-.08,0)--(.35,0);
    \draw[xir icon,draw=xirActorColor,-{Latex[length=1.3mm]}] (-.08,0)--(.28,-.25);
    \fill[xirActorColor] (-.10,0) circle (.055);
  }},
  pics/xir-relayer/.style={code={
    \draw[xir icon,draw=xirActorColor] (-.18,-.27)--(0,.08)--(.18,-.27) (-.11,-.13)--(.11,-.13);
    \fill[xirActorColor] (0,.08) circle (.045);
    \draw[xir icon,draw=xirActorColor] (-.12,.13) arc[start angle=145,end angle=35,radius=.15];
    \draw[xir icon,draw=xirActorColor] (-.23,.18) arc[start angle=145,end angle=35,radius=.28];
  }},
  pics/xir-evidence/.style={code={
    \path[draw=xirControlColor,fill=xirControlColor!7,xir icon]
      (-.31,-.34)--(-.31,.34)--(.13,.34)--(.31,.16)--(.31,-.34)--cycle;
    \draw[xir icon,draw=xirControlColor] (.13,.34)--(.13,.16)--(.31,.16);
    \draw[xir icon,draw=xirControlColor] (-.18,-.03)--(-.05,-.16)--(.19,.10);
  }},
  pics/xir-receipt/.style={code={
    \path[draw=xirProtocolColor,fill=xirProtocolColor!7,xir icon]
      (-.34,-.34)--(-.34,.34)--(.11,.34)--(.30,.15)--(.30,-.34)--cycle;
    \draw[xir icon,draw=xirProtocolColor] (.11,.34)--(.11,.15)--(.30,.15);
    \draw[xir icon,draw=xirProtocolColor] (-.18,.08)--(.13,.08) (-.18,-.04)--(.13,-.04);
    \draw[xir icon,draw=xirProtocolColor] (-.08,-.19) circle (.10) (.10,-.19) circle (.10);
  }},
  pics/xir-policy/.style={code={
    \draw[xir icon,draw=xirControlColor,fill=xirControlColor!7,rounded corners=.025cm]
      (-.34,-.32) rectangle (.34,.30);
    \draw[xir icon,draw=xirControlColor] (-.13,.35) rectangle (.13,.25);
    \draw[xir icon,draw=xirControlColor] (-.23,.11)--(-.17,.04)--(-.08,.15) (.02,.10)--(.23,.10);
    \draw[xir icon,draw=xirControlColor] (-.23,-.10)--(-.17,-.17)--(-.08,-.06) (.02,-.11)--(.23,-.11);
  }},
  pics/xir-replay/.style={code={
    \draw[xir icon,draw=xirGatewayColor,fill=xirGatewayColor!8,rounded corners=.03cm]
      (-.31,-.28) rectangle (.31,.10);
    \draw[xir icon,draw=xirGatewayColor] (-.18,.10)--(-.18,.22)
      arc[start angle=180,end angle=0,x radius=.18,y radius=.18]--(.18,.10);
    \fill[xirGatewayColor] (0,-.05) circle (.045);
    \draw[xir icon,draw=xirGatewayColor] (0,-.08)--(0,-.18);
  }},
  pics/xir-network/.style={code={
    \draw[xir icon,draw=xirAppColor,fill=xirAppColor!6,rounded corners=.02cm]
      (-.40,-.31) rectangle (.40,.31);
    \draw[xir icon,draw=xirAppColor] (-.40,.10)--(.40,.10) (-.40,-.11)--(.40,-.11);
    \fill[xirAppColor] (-.27,.205) circle (.035) (-.27,0) circle (.035) (-.27,-.205) circle (.035);
    \draw[xir icon,draw=xirAppColor] (-.13,.205)--(.27,.205) (-.13,0)--(.27,0) (-.13,-.205)--(.27,-.205);
  }},
  pics/xir-result/.style={code={
    \draw[xir icon,draw=xirResultColor,fill=xirResultColor!9] (0,0) circle (.34);
    \draw[xir icon,draw=xirResultColor,line width=.9pt] (-.17,-.01)--(-.05,-.14)--(.19,.14);
  }}
}

%% file: generated/paper-values.tex
\newcommand{\DirectRate}{14.64}

\newcommand{\ProtocolSwitchingRate}{96.86}

\expandafter\newcommand\csname HHMedian\endcsname{7.103}
\expandafter\newcommand\csname HHP95\endcsname{12.157}
\expandafter\newcommand\csname HHP99\endcsname{17.240}
\expandafter\newcommand\csname HHCoordinatorTransactionsTotal\endcsname{10,000}
\expandafter\newcommand\csname HHCoordinatorGasTotal\endcsname{941,999,905}
\expandafter\newcommand\csname HHCoordinatorCalldataTotal\endcsname{4,360,000}
\expandafter\newcommand\csname LLMedian\endcsname{5.159}
\expandafter\newcommand\csname LLP95\endcsname{6.154}
\expandafter\newcommand\csname LLP99\endcsname{7.180}
\expandafter\newcommand\csname LLCoordinatorTransactionsTotal\endcsname{10,000}
\expandafter\newcommand\csname LLCoordinatorGasTotal\endcsname{1,919,916,016}
\expandafter\newcommand\csname LLCoordinatorCalldataTotal\endcsname{4,680,000}
\expandafter\newcommand\csname HLMedian\endcsname{10.955}
\expandafter\newcommand\csname HLP95\endcsname{14.893}
\expandafter\newcommand\csname HLP99\endcsname{19.985}
\expandafter\newcommand\csname HLCoordinatorTransactionsTotal\endcsname{50,000}
\expandafter\newcommand\csname HLCoordinatorGasTotal\endcsname{7,219,771,195}
\expandafter\newcommand\csname HLCoordinatorCalldataTotal\endcsname{54,440,000}
\expandafter\newcommand\csname LHMedian\endcsname{10.918}
\expandafter\newcommand\csname LHP95\endcsname{14.011}
\expandafter\newcommand\csname LHP99\endcsname{19.079}
\expandafter\newcommand\csname LHCoordinatorTransactionsTotal\endcsname{50,000}
\expandafter\newcommand\csname LHCoordinatorGasTotal\endcsname{7,287,656,584}
\expandafter\newcommand\csname LHCoordinatorCalldataTotal\endcsname{56,360,000}

%% file: sections/1-introduction.tex
\section{Introduction}
\label{sec:intro}

Blockchain applications increasingly use multiple networks to access assets,
liquidity, computation, and users.  Cross-chain protocols provide the
communication layer between these networks.  A source endpoint sends an
application cross-chain message, a verification mechanism authenticates the source event,
and a destination endpoint executes the requested operation
~\cite{hardjono2018design,zamyatin2021sok,belchior2021survey}.  Protocols use
different trust models and endpoint interfaces~\cite{belchior2021survey}.
Related approaches include zkBridge~\cite{xie2022zkbridge}, cross-chain
integrated execution~\cite{yin2025atomic}, SightCVC~\cite{yang2025sightcvc},
and MAP~\cite{cao2025map}.

The point-to-point configuration model makes direct communication depend on
an explicitly configured source--destination pair. LayerZero OApps~\cite{layerzerooapp}
register a trusted peer for an endpoint identifier, Hyperlane
routers~\cite{hyperlanerouter} register remote domains, and other protocols
register emitters~\cite{wormholevaa}, receivers, or equivalent
endpoints~\cite{ccipofframp}.
For an application on blockchain $A$ to reach blockchain $C$ directly, this
model requires an $A\!\to\!C$ configuration. Reusing existing
$A\!\to\!B$ and $B\!\to\!C$ connections instead requires an intermediate
handoff that preserves the application operation and the verification results
from both deliveries. When the connections use different protocols, the
handoff must also bridge their message formats and verifier interfaces.

The extent of this opportunity appears in activity feeds from
Axelar~\cite{axelarscanapi}, CCIP~\cite{ccipatlasfeed},
Hyperlane~\cite{hyperlanegraphql}, LayerZero~\cite{layerzeroscanapi},
Relay~\cite{relayapi}, and Wormhole~\cite{wormholescanapi}.
After resolving network identities and retaining mainnet events between
different blockchains, the January--October 2025 snapshot contains
24,983,410 events from 286 active blockchains. Its 15,965 protocol-labeled
directed edges represent 11,935 directly connected ordered pairs, or
\DirectRate\% of the $286\times285=81{,}510$ possible pairs.
Full direct connectivity through one protocol would require 81,510
point-to-point configurations. Composing the observed connections offers a
way to reach additional pairs through intermediate blockchains.

IBC multi-hop channels~\cite{ibcmultihop} and xRoute~\cite{xroute2026}
support path composition within IBC. Axelar General Message
Passing~\cite{axelargmp} exposes a common messaging service, while
service-mesh interoperability~\cite{kapsoulis2026mesh} coordinates cross-chain
transactions through mesh services.  The remaining problem is to combine
existing direct protocol connections into same-protocol and cross-protocol
paths while preserving verification and execution across protocol boundaries.

Combining direct connections across protocol boundaries creates two research
challenges.

\textbf{Challenge 1: verification continuity.}  Each protocol
authenticates its own messages using a specific verifier configuration and
evidence format.  The final protocol authenticates the last delivery. Establishing the earlier
verification history requires each preceding result to be bound to the same
operation and path. For example, two individually valid receipts from
different executions must be rejected when presented as one history.

\textbf{Challenge 2: execution continuity.}  Protocols use different message
formats, endpoint identifiers, callbacks, and replay checks.  A message that
leaves one protocol must retain the same application identity, destination,
path position, and replay identity when it enters the next protocol.

To address these challenges, this paper proposes XIR, a framework that composes existing cross-chain protocol connections through a verifiable intermediate representation.
This representation binds one application cross-chain message to an ordered record of authenticated protocol deliveries, preserving verification evidence across protocol boundaries.
XIR comprises XIR Gateways that manage this record, XIR Adapters that connect protocol interfaces, an XIR Registry that stores configurations, and an XIR Router that selects paths.
On the source blockchain, an XIR Gateway authenticates the application and creates the record.
On each intermediate blockchain, an XIR Gateway verifies the incoming delivery and extends the record.
On the destination blockchain, an XIR Gateway verifies the complete history before execution.
These operations provide verification continuity (Challenge~1) by linking each delivery to the preceding history, and execution continuity (Challenge~2) by maintaining stable message identity across protocols.
Under the stated assumptions, the analysis establishes source authorization, ordered evidence integrity, per-hop policy enforcement, and at-most-once delivery.

With XIR Gateways and XIR Adapters at the required intermediate
blockchains, same-protocol composition over the cumulative graph reaches
46,187 ordered pairs. Allowing protocol switches expands this set to 78,953
pairs, or
\ProtocolSwitchingRate\% of the 81,510-pair maximum.  Hyperlane has the
largest same-protocol closure at 27,059 pairs (33.20\%).
Figure~\ref{fig:protocol-xir-reachability} shows how the two forms of
composition expand reachability over the existing direct connections.

\begin{figure*}[t]
\centering
\includegraphics[width=\textwidth]{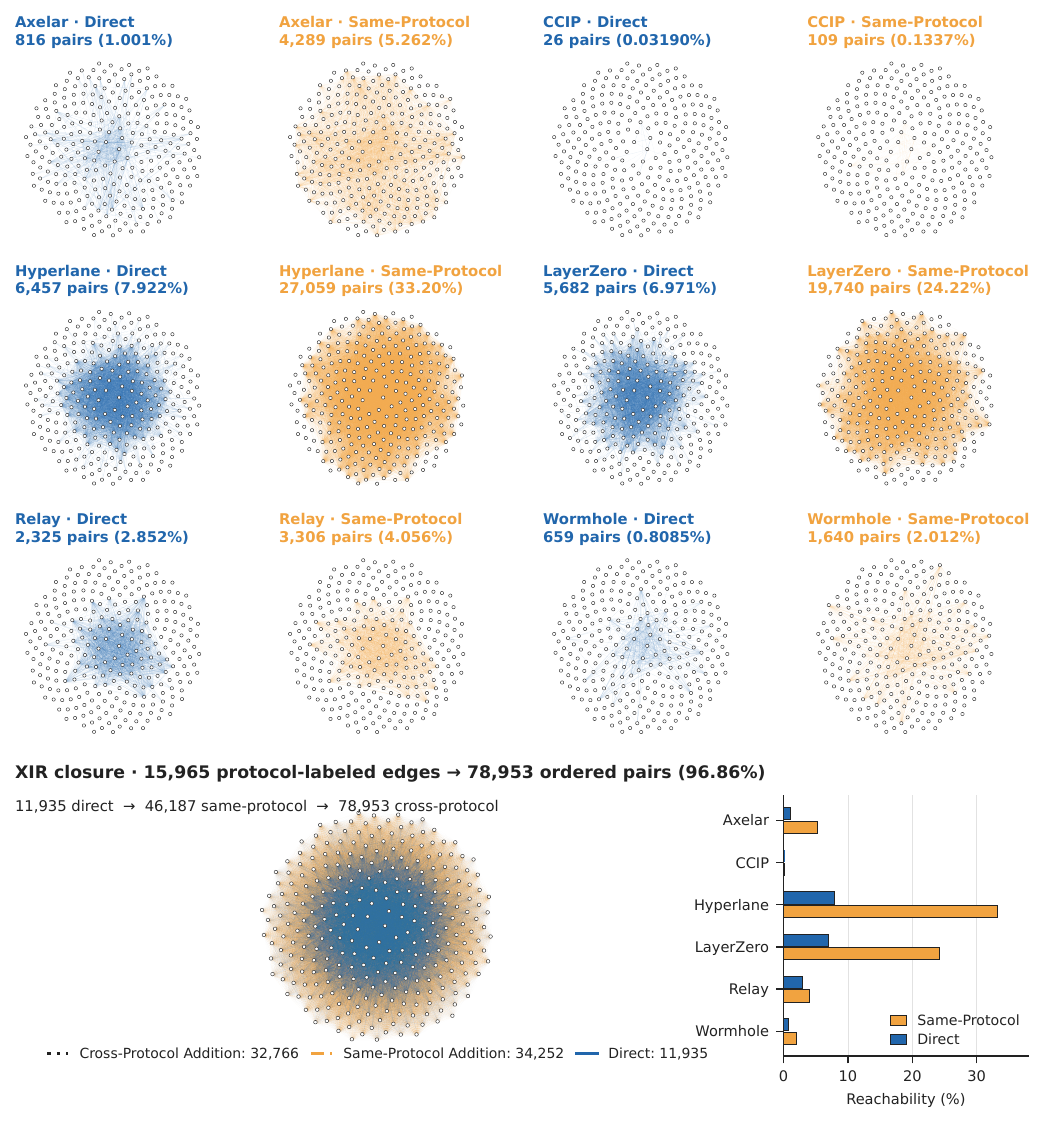}
\caption{\textbf{Reachability in the measured 286-blockchain graph.} The 15,965 protocol-labeled directed edges connect 11,935 pairs directly. Same-protocol composition reaches 46,187 pairs, and cross-protocol composition reaches 78,953 pairs.}
\label{fig:protocol-xir-reachability}
\end{figure*}

Figure~\ref{fig:xir-user-workflow} next compares direct protocol configuration
with XIR path composition from the application user's perspective.

Reproducing the 78,953-pair target with direct connections requires 78,953
point-to-point protocol configurations, including 67,018 configurations for
pairs that are not directly connected in the measured graph.  XIR composes the
15,965 existing protocol connections through one XIR Gateway placement on
each of the 286 blockchains.  The observed protocol coverage also
contains 493 blockchain--protocol integrations, one for each protocol
available on each blockchain.  Point-to-point configurations, XIR Gateway deployments, and
blockchain--protocol integrations are separate units.  Section~\ref{sec:empirical-cost-analysis}
reports these architecture counts separately and compares deployment cost in
gas.

The XIR prototype integrates Hyperlane and LayerZero, and its implementation
is open source.\footnote{XIR implementation:
\url{https://github.com/lysrain21/XIR}.}
It completes 40,000 local deliveries and 780 adversarial case-runs, with seven
successful deliveries recorded on public testnets. A separate campaign of
22,000 executions measures the cost of verification history and multi-hop
paths. Section~\ref{sec:implementation} presents the experimental settings,
cost models, and results.

This paper develops a framework for composing cross-chain protocols while
preserving message identity and verification history. Its main contributions
are summarized as follows.

\begin{itemize}
  \item \textbf{Verifiable intermediate representation.} XIR binds the application cross-chain
  message to an ordered record of authenticated protocol
  deliveries, enabling the destination to verify evidence accumulated across
  protocol boundaries. Sections~\ref{sec:model} and~\ref{sec:design} define
  the representation and its verification procedure.

  \item \textbf{Protocol composition with explicit guarantees.} XIR
  Gateways and XIR Adapters reuse configured connections for same-protocol and
  cross-protocol paths while preserving the application record.
  Section~\ref{sec:analysis} establishes source authorization, ordered
  evidence integrity, per-hop policy enforcement, and at-most-once delivery
  under the stated assumptions.

  \item \textbf{Open-source implementation and evaluation.} 
  We developed an XIR prototype system in Solidity with a stateful Python orchestrator, integrated with Hyperlane and LayerZero cross-chain protocols. Based on an analysis of approximately 25 million cross-chain transaction events, we evaluated the reachability of XIR and validated its functionality and overhead. Experimental results show that XIR achieves 96.86\% reachability on the cross-chain transaction graph. Compared with direct point-to-point connections serving the same reachable pairs, XIR avoids 67,018 additional point-to-point protocol configurations. Sections~\ref{sec:empirical-cost-analysis} and~\ref{sec:implementation} present these results.
\end{itemize}


%% file: sections/2-background.tex
\section{Background and Problem Formulation}
\label{sec:background}

\subsection{Cross-Chain Protocols and Protocol Interoperability}

Blockchain interoperability allows different blockchains to exchange and use
information
~\cite{hardjono2018design,belchior2021survey}.  A cross-chain protocol
provides three common functions: a source endpoint emits an application
cross-chain message, a verification mechanism authenticates the source event, and a
relayer, executor, or solver submits the cross-chain message for
execution at the destination~\cite{zamyatin2021sok,falazi2024crosschain}.
Table~\ref{tab:protocol-comparison} maps these functions to the six protocols
in the measured dataset and to XIR.

\begin{table*}[t]
\centering
\caption{Cross-chain protocol functions and the corresponding XIR components.}
\label{tab:protocol-comparison}
\scriptsize
\setlength{\tabcolsep}{3pt}
\renewcommand{\arraystretch}{1.06}
\begin{tabular}{@{}p{0.09\textwidth}p{0.26\textwidth}p{0.28\textwidth}p{0.29\textwidth}@{}}
\toprule
System & Message verification & Delivery and execution & Endpoint configuration \\
\midrule
Axelar~\cite{axelar} & Validator-authorized Gateway commands
& Relayer submission and \texttt{AxelarExecutable} callback
& Chain, Gateway, and destination identifiers~\cite{axelargatewayinterface} \\
CCIP v1.6~\cite{ccip} & Oracle network running the Commit plugin
& Executing plugin, OffRamp, and receiver callback
& Router, OnRamp/OffRamp, lane, and receiver~\cite{ccipofframp} \\
Hyperlane~\cite{hyperlane} & Recipient-selected Interchain Security Module
& Relayer calls \texttt{Mailbox.process}~\cite{hyperlanemailbox}
& Domain and remote router~\cite{hyperlanerouter}; hook and ISM~\cite{hyperlane} \\
LayerZero~\cite{layerzero} & Decentralized Verifier Networks and Message Library
& Executor, Endpoint, and \texttt{lzReceive}
& EID, peer, Message Library, DVN, and Executor~\cite{layerzerooapp} \\
Relay~\cite{relayprotocol} & Oracle attestations of deposits and fills
& Solver fill and Hub settlement
& Depository, Hub, solver, and target action \\
Wormhole~\cite{wormholevaa} & Guardian-signed~\cite{wormholeguardians} Verified Action Approval
& Relayer submission and receiver execution
& Trusted emitter, receiver, and consistency level~\cite{wormholevaa} \\
XIR & Ordered verification evidence from cross-chain protocols
& XIR Gateway forwarding and application execution
& Existing protocol configurations, XIR Gateways, and XIR Adapters \\
\bottomrule
\end{tabular}
\end{table*}

Protocol interoperability allows different cross-chain protocols to carry one
message along a pair-to-pair path.  Here, \emph{pair-to-pair} refers to
communication from the source blockchain to the destination blockchain,
including any intermediate hops.  At an intermediate blockchain, an XIR Gateway receives a
message through one XIR Adapter and sends the same application operation
through another.  This step is a \emph{handoff}.  A same-protocol handoff uses
the same protocol for the incoming and outgoing connections.  A
\emph{protocol switch} uses different protocols.  XIR records the incoming
and outgoing protocols at each handoff.  Each connection retains its protocol
endpoint, point-to-point configuration, and verifier.

Figure~\ref{fig:xir-user-workflow} shows the interaction from the application
user's perspective.  The application submits one request with a destination
and a policy.  Direct delivery requires an $A\!\to\!C$ protocol configuration.
XIR instead selects existing $A\!\to\!B$ and $B\!\to\!C$ connections and returns
the execution status from $C$.  The two connections may use the same protocol
or different protocols.

\begin{figure*}[t]
\centering
\includegraphics[width=\textwidth]{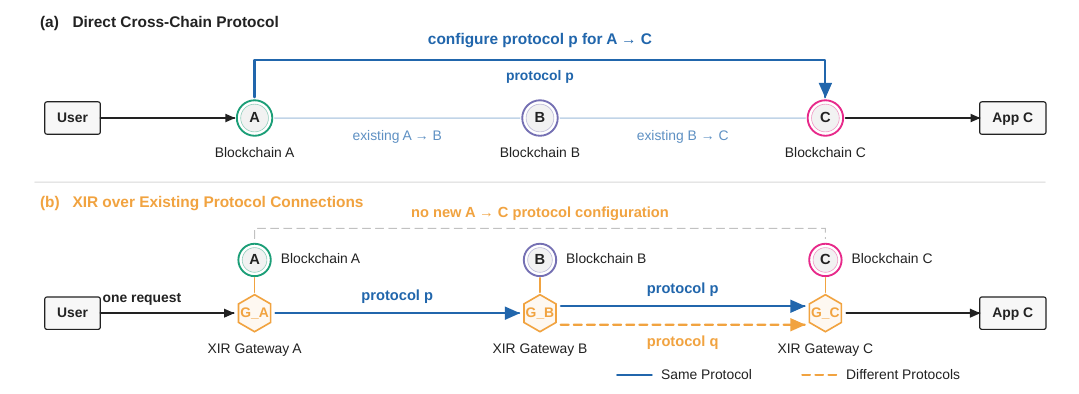}
\caption{\textbf{Cross-chain communication with and without XIR.} Direct delivery requires an $A\!\to\!C$ protocol configuration. XIR carries the request over existing $A\!\to\!B$ and $B\!\to\!C$ connections, using the same protocol ($p\!\to\!p$) or different protocols ($p\!\to\!q$). Ring color marks the source, intermediate, and destination blockchain.}
\label{fig:xir-user-workflow}
\end{figure*}

\subsection{Reachability across Blockchains}

The connectivity dataset contains 24,983,410 distinct events from Axelar, CCIP, Hyperlane, LayerZero,
Relay, and Wormhole between 1 January 2025 (inclusive) and 1 November 2025
(exclusive). Collection used protocol explorer and service APIs: AxelarScan,
the Atlas CCIP feed, Hyperlane Explorer, LayerZero Scan, Relay, and
WormholeScan. Normalization resolves chain aliases to canonical network
identities, retains distinct mainnet endpoints, and deduplicates by protocol
and event identifier. The release contains 60 protocol/month Parquet files
with 18 fields, including event identifiers, UTC timestamps, original and
canonical endpoint identities, available transaction hashes, and source-file
provenance. It also includes the edge table, identity registry, collection
limitations, checksums, and reproduction scripts.
The dataset is publicly available on Kaggle under CC BY 4.0
(version 4, DOI: \href{https://doi.org/10.34740/kaggle/dsv/19934883}{10.34740/kaggle/dsv/19934883}).

The records describe captured activity. Coverage is incomplete: the APIs limit queries, the collection kept no query logs, and the CCIP feed is particularly sparse.
Each timestamp follows the provider's event definition, and a row records a source-side event without confirming destination execution. An absent edge does not
establish that a protocol or direct connection was unavailable.

Reachability is a standard property of directed graphs: a vertex $v$ is
reachable from $u$ if a directed path connects $u$ to $v$
~\cite{bangjensen2009digraphs}.  Let $G=(V,E,\lambda)$ be a protocol-labeled
directed multigraph.  Each vertex represents a blockchain, each edge
represents an observed direct protocol connection, and $\lambda(e)\in\mathcal P$
identifies the edge protocol.  Let $s(e)$ and $t(e)$ denote the source and
destination vertices of edge $e$.  For protocol $p$, define
\begin{equation}
E_p=\{(s(e),t(e))\mid e\in E\land\lambda(e)=p\}.
\label{eq:protocol-edge-relation}
\end{equation}

Let $\Delta_V=\{(v,v)\mid v\in V\}$ be the set of self pairs.
The direct, same-protocol, and cross-protocol reachability relations are
\begin{align}
\mathcal R_{\mathrm{direct}} &= \left(\bigcup_{p\in\mathcal P}E_p\right)\setminus\Delta_V,\nonumber\\*
\mathcal R_{\mathrm{same}} &= \left(\bigcup_{p\in\mathcal P}E_p^{+}\right)\setminus\Delta_V,\nonumber\\*
\mathcal R_{\mathrm{XIR}} &= \left(\bigcup_{p\in\mathcal P}E_p\right)^{+}\setminus\Delta_V,
\label{eq:reachability-relations}
\end{align}
where $(\cdot)^{+}$ denotes transitive closure.  The last relation describes composition over the cumulative observed graph,
with XIR Gateways and XIR Adapters at the intermediate blockchains
required by each path.  For any relation $\mathcal R$ containing only non-self
ordered pairs, the reachability ratio is
\begin{equation}
\rho(\mathcal R)=\frac{|\mathcal R|}{|V|(|V|-1)}.
\label{eq:reachability-ratio}
\end{equation}

For $|V|=286$, the denominator is $286\times285=81{,}510$.  The observed
relations give
\begin{align}
\rho(\mathcal R_{\mathrm{direct}})&=\frac{11{,}935}{81{,}510}=14.64\%,\nonumber\\
\rho(\mathcal R_{\mathrm{same}})&=56.66\%,\qquad
\rho(\mathcal R_{\mathrm{XIR}})=96.86\%.
\label{eq:observed-reachability-results}
\end{align}
Same-protocol composition adds 34,252 pairs beyond direct connections.
Protocol switches add another 32,766 pairs.
These values describe potential reachability over connections observed during
the ten-month window. Each edge records activity within that window, and
composition assumes XIR Gateway and XIR Adapter coverage along the path.
Temporal snapshots show how the result changes with the observation window:
monthly protocol-switching reachability ranges from 46.4\% to 65.7\%, and the
final rolling-90-day union reaches 74.2\%. The cumulative reference is 96.9\%.
All three use the same 286-blockchain population. Raising the minimum activity
per directed protocol edge from one to 1,000 distinct events reduces the
cumulative protocol-switching result from 78,953 to 18,913 pairs. These
analyses describe the dependence of reachable paths on observation duration
and edge activity.

For $N=|V|$ blockchains, direct connectivity requires one point-to-point
protocol configuration for every ordered pair, or $N(N-1)=O(N^2)$ configurations.
XIR uses one XIR Gateway deployment per blockchain, giving $N=O(N)$ deployments.  It also requires one
integration for each incident blockchain--protocol pair.  Let $V_p$ be the set
of blockchains incident to an edge labeled $p$.  Then
$I_{\mathrm{XIR}}=\sum_{p\in\mathcal P}|V_p|=O(N)$ integrations are required for a fixed protocol set. Applied to the
measured graph, this architecture requires 286 XIR Gateway placements and 493
blockchain--protocol integrations.  Section~\ref{sec:empirical-cost-analysis} reports these units
separately and compares deployment gas.

\subsection{Hop Count}
\label{sec:hop-count}

One hop is a directed cross-chain protocol connection between adjacent blockchains,
consistent with the terminology of multi-hop IBC~\cite{ibcmultihop}.  A path
\begin{equation}
P=(v_0,e_1,v_1,\ldots,e_h,v_h)
\label{eq:path-and-hops}
\end{equation}
has $h$ hops and $h-1$ intermediate blockchains.  Both
$A\xrightarrow{p}B\xrightarrow{p}C$ and
$A\xrightarrow{p}B\xrightarrow{q}C$ contain two hops.
A protocol switch performed by an XIR Gateway at an intermediate blockchain does not
add a hop; it is counted separately.

Hop count and protocol-switch count describe different path properties.
Writing $p_i=\lambda(e_i)$, the number of protocol switches is
\begin{equation}
c(P)=\sum_{i=2}^{h}\mathbf{1}[p_i\ne p_{i-1}].
\label{eq:protocol-change-count}
\end{equation}
The first example has $c(P)=0$, and the second has $c(P)=1$.  Experiments 2 and
3 use this distinction to separate the number of earlier protocol deliveries
from the number of protocol switches.

Among the 78,953 finite ordered pairs reached by cross-protocol composition in
the measured graph, the shortest-path hop distribution is
11,935 at one hop (15.12\%), 59,711 at two hops (75.63\%), 7,186 at three hops
(9.102\%), and 121 at four hops (0.1533\%).
The cumulative percentages are 15.12\%, 90.75\%, 99.85\%, and 100.0\%.
Nearest-rank $P_{50}$, $P_{95}$, and $P_{99}$ are two, three, and three hops,
respectively.  The directed diameter over these finite pairs is four hops and
is attained by 121 ordered pairs; the other 2,557 ordered pairs are
unreachable.  Figure~\ref{fig:xir-hop-distribution} shows the cumulative
distribution.  Paths of one to four hops span the complete cross-protocol
shortest-path range.

\begin{figure}[t]
\centering
\includegraphics[width=\columnwidth]{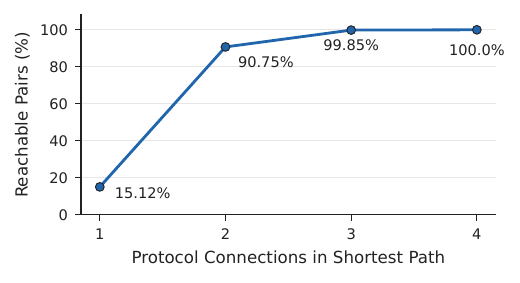}
\caption{\textbf{Shortest paths in the measured graph.} The curve shows the
cumulative fraction of the 78,953 reachable ordered pairs.  All shortest paths
use at most four protocol connections.}
\label{fig:xir-hop-distribution}
\end{figure}

Individual same-protocol closures can contain longer paths: CCIP has nine
reachable pairs whose shortest same-protocol path has five hops.  The
concentration at two hops motivates the three-blockchain topology in
Experiment 1, and the cross-protocol range motivates the path lengths in
Experiment 3.

Existing same-protocol systems also reuse connections through path composition.
IBC ICS-033~\cite{ibcmultihop} chains proofs across existing IBC connections,
and xRoute~\cite{xroute2026} adds policy-aware selection to IBC routes.
Axelar~\cite{axelargmp} exposes a common messaging interface across its
supported networks, while service-mesh work~\cite{kapsoulis2026mesh}
coordinates cross-chain transactions through mesh services. XIR composes direct protocol
connections across protocol domains and carries the accepted evidence into a
single destination decision.

\subsection{The Verification Problem}

Cross-chain protocols authenticate messages through different authorities and
evidence formats, including validator committees, light clients, oracle
networks, configurable verifier sets, and signed attestations
~\cite{zamyatin2021sok,notland2026sok}.  A one-hop destination observes the
evidence authenticated by its protocol endpoint.  A cross-protocol path adds an
intermediate handoff whose inbound verifier belongs to another protocol.

The destination must bind every receipt to the same application operation and
ordered path.  It must also bind each receipt to the protocol verification
result, the selected verifier configuration, and the next handoff.  The
binding covers the payload and replay identity.  Analyses of bridge failures
show the importance of endpoint
authentication, configuration integrity, and inherited trust assumptions
~\cite{lee2023bridgehacks,zhang2024bridgesecurity}.

Related work also addresses atomic asset exchange and transaction scheduling.
Chuchu~\cite{zhuo2026chuchu} coordinates cross-chain swap completion and
refunds through asset locking state transitions and linked hashlock groups.
Concordia~\cite{liu2026concordia} shares predicted account-access signals
among shards to guide transaction packing and reduce conflicts in a sharded
blockchain.

Systems that strengthen verification improve one protocol boundary.
zkBridge~\cite{xie2022zkbridge} proves remote consensus state, while
CAVer~\cite{guo2025caver} uses cross-chain-specific accumulators to reduce
message-proof size and destination verification overhead.  XIR records the verification
evidence accepted from each protocol, links the evidence in path order, and
authenticates the complete record at the destination.  This representation
supports cross-protocol multi-hop execution while preserving the verifier of
each protocol connection.

The task is to compose existing protocol connections while maintaining
\emph{verification continuity} and \emph{execution continuity}.  The former
binds the verifier sequence and its configuration.  The latter binds the
stable application operation, payload, path state, and replay identity.
Sections~\ref{sec:model} and~\ref{sec:design} introduce the XIR components and
protocol steps that realize these properties.

%% file: sections/3-system-model.tex
\section{System Model}
\label{sec:model}

XIR combines existing direct protocol connections while preserving one
application operation and its verified execution history.  The running example
is $A\xrightarrow{p}B\xrightarrow{q}C$, where protocols $p$ and $q$ have
active endpoints and point-to-point configurations, as in Hyperlane
routers~\cite{hyperlanerouter} and LayerZero OApps~\cite{layerzerooapp}.

\subsection{Entities, Roles, and State}

An application submits a cross-chain message, destination, and security requirement to a
source \emph{XIR Gateway}.  An XIR Gateway connects an application to XIR on
one blockchain.  The source XIR Gateway authenticates the caller and creates the
XIR operation.  An intermediate XIR Gateway verifies one protocol delivery and
starts the next.  The destination XIR Gateway verifies the complete path before
invoking the receiver.

The \emph{XIR Router} reads the \emph{XIR Registry} and returns an admissible
path.  The XIR Registry stores protocol endpoints, point-to-point configurations,
XIR Adapters, verification profiles, validity windows, security levels, and
approved incoming XIR Adapters.  An \emph{XIR Adapter} connects XIR to one
cross-chain protocol.  Its incoming side records an authenticated protocol
callback.  Its outgoing side verifies that record before sending the next
message.  Relayers submit messages, and a root signer certifies a finalized
source event.

The XIR Registry maintains versioned configuration.  XIR Adapters store accepted
evidence and approved prior verifiers.  XIR Gateways maintain application
nonces and consumed message identifiers.  Route selection does not change this
state.  Message acceptance uses the on-chain configuration and authenticated
XIR Adapter state.

Figure~\ref{fig:xir-internal-workflow} summarizes how these components execute
one request.  The source creates one application record, the XIR Router selects
configured protocol connections, each protocol verifies its own delivery, and
the destination verifies the ordered record before invoking the application.

\begin{figure*}[t]
\centering
\includegraphics[width=\textwidth]{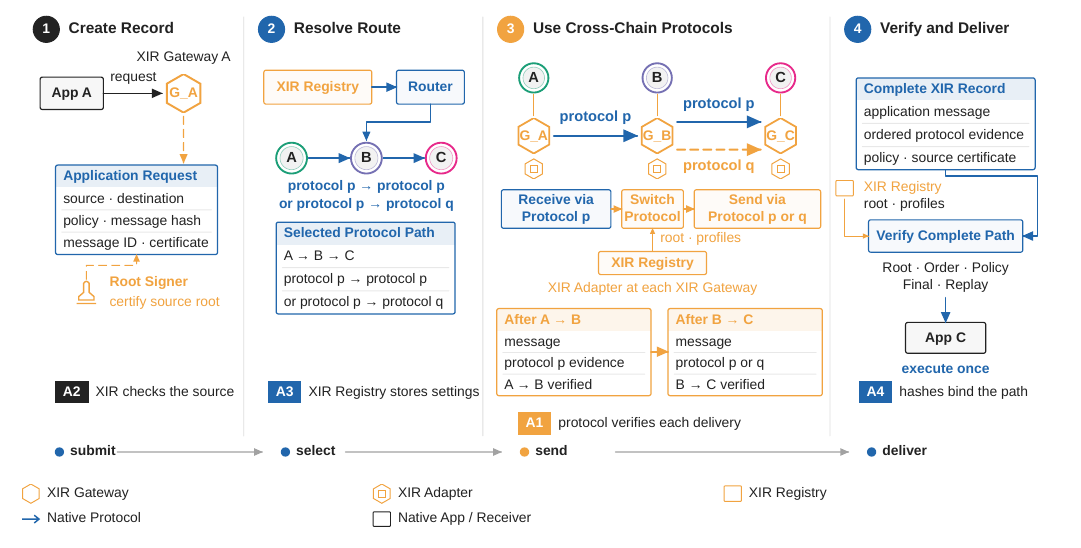}
\caption{\textbf{XIR workflow.} The XIR Router selects configured protocol connections. XIR Gateways and XIR Adapters record the verification result of each delivery. The destination verifies the ordered record before invoking the application. Ring color marks the source, intermediate, and destination blockchain.}
\label{fig:xir-internal-workflow}
\end{figure*}

\subsection{Records, Identifiers, and Path Execution}

XIR carries the following verifiable intermediate representation (V-IR):
\begin{equation}
X=\langle R,ctx,cert,[C_1,\ldots,C_h]\rangle ,
\label{eq:xir-envelope-overview}
\end{equation}
The application record $R$ identifies the operation, and $ctx$ carries its
security requirement and policy commitment. The certificate $cert$ authorizes the
source record. Each receipt $C_i$ identifies an authenticated delivery over
protocol connection $i$. The source also derives a root identifier $rid$ and
a message identifier $mid$: the former binds the authorized record, and the
latter identifies its destination execution. Section~\ref{sec:design} defines
these encodings.

A path is a directed sequence
$v_0\xrightarrow{p_1}v_1\cdots\xrightarrow{p_h}v_h$ in the XIR Registry view. Each
edge contributes one hop.  A protocol switch occurs when $p_i\ne p_{i+1}$.
The handoff at the intermediate XIR Gateway does not add a hop.  Verified evidence
from connection $i$ produces $C_i$ and extends a commitment to the path
history.  The XIR Gateway then uses an XIR Adapter to continue through the same
protocol or switch to another protocol.

\subsection{Adversary, Trust, and Required Properties}
\label{sec:model-assumptions}

An adversary controlling a submitter or the communication network may delay,
duplicate, reorder, or alter unauthenticated inputs; replay a completed
request; splice evidence from another execution; or present an expired
profile.  The destination accepts a message only when its application record
is authorized and every protocol delivery satisfies the selected policy.

The analysis uses four assumptions. \textbf{A1, protocol-delivery soundness:}
an authenticated protocol callback binds its endpoint, remote XIR Adapter, message,
and evidence identifier to source blocks that the protocol's configured finality rule treats as final. This condition applies to every hop, including deliveries that start on an intermediate blockchain.
\textbf{A2, XIR-component soundness:} XIR Gateways, XIR Adapters, and the root
signer execute the specified checks. The root signer certifies only source
events in blocks that the source blockchain's configured finality rule treats
as final; certified source roots are not rolled back. At the destination, the
consumed entry and the application effect are committed in the same transaction
and removed together on rollback; the proof concerns the final destination
state after rollback. \textbf{A3,
configuration authenticity:} the XIR Registry contains the intended roots,
profiles, XIR Adapters, endpoint configurations, validity windows, and approved
incoming XIR Adapters.
\textbf{A4, cryptographic binding:} canonical encodings are injective,
domain-separated hashes resist collisions, and root signatures resist forgery.

Under A1--A4, destination acceptance provides four properties.
\textbf{G1, source authorization:} the accepted record and context match the
certified source root. \textbf{G2, ordered evidence integrity:} the receipt
sequence is endpoint-continuous and bound to one final protocol delivery.
\textbf{G3, policy enforcement:} every receipt resolves to an active profile
whose security level meets $ctx$. \textbf{G4, at-most-once delivery:} one
message identifier commits at most one destination effect.
Section~\ref{sec:analysis} establishes them under A1--A4.

Combining protocol connections raises the two challenges introduced in
Section~\ref{sec:intro}.  \emph{Verification continuity} requires separate
protocol callbacks to form one ordered history that the destination can
verify.  The callback for the last connection authenticates the final delivery.
XIR carries earlier verification results forward by linking one receipt to
each accepted delivery and binding the complete sequence to that final message.  \emph{Execution continuity}
requires one application operation to retain its identity across different
protocol formats and callbacks.  XIR keeps $R$, $ctx$, $rid$, and $mid$ stable
while XIR Gateways and XIR Adapters perform each handoff.

\subsection{Pair-to-Pair Workflow}

Algorithm~\ref{alg:xir-route-execute} details the workflow in
Figure~\ref{fig:xir-internal-workflow}. The inputs $v_{\mathrm{src}}$ and $v_{\mathrm{dst}}$ identify the
source and destination blockchains, $a_{dst}$ identifies the destination
application, $m$ is the message payload, and $\pi$ specifies the security
policy. The source XIR Gateway authenticates the caller and creates $R$, storing
$a_{dst}$ as $R.destinationApp$.

The XIR Router returns a path $P=(v_0,e_1,\ldots,e_h,v_h)$ and its verification
profiles $\Phi=(\Phi_1,\ldots,\Phi_h)$. Here, $G_i$ is the registered XIR Gateway
on blockchain $v_i$. Let $T_i=[C_1,\ldots,C_i]$ denote the history after $i$
deliveries. The algorithm maintains this list in $T$ and its commitment $q_i$
in $prefix$, starting with $T_0=[\ ]$ and $q_0$. Each iteration represents an
asynchronous protocol delivery and its authenticated reception. Route selection
uses one XIR Registry view; each on-chain step checks the configuration available
when it executes. A failed assertion aborts the current operation. The final
consume-and-invoke step is atomic within the destination transaction.

\floatname{algorithm}{Algorithm}
\begin{algorithm}[t]
\footnotesize
\caption{XIR path composition and delivery}
\label{alg:xir-route-execute}
\begin{algorithmic}[1]
\Require source blockchain $v_{\mathrm{src}}$, destination blockchain $v_{\mathrm{dst}}$,
\Statex \hspace{\algorithmicindent}destination application $a_{dst}$, payload $m$, policy $\pi$
\Ensure execution result on successful completion
\State $(R,ctx,rid,mid)\gets\Call{CreateRecord}{v_{\mathrm{src}},v_{\mathrm{dst}},a_{dst},m,\pi}$
\State $cert\gets\Call{CertifyFinalizedRoot}{rid}$
\State $(P,\Phi)\gets\Call{ResolvePath}{v_{\mathrm{src}},v_{\mathrm{dst}},ctx}$
\State \textbf{assert} $P\ne\varnothing$
\State $h\gets |P.edges|$
\State $(G_0,\ldots,G_h)\gets\Call{ResolveGateways}{P}$
\State $T\gets[\ ]$
\State $prefix\gets H(\mathtt{XIR\_ROOT\_V1}\parallel rid)$
\For{$i\gets1$ to $h$}
  \State $e_i\gets P.edges[i]$
  \State \textbf{assert} \Call{ValidProfile}{$\Phi_i,G_{i-1},G_i,ctx$}
  \State $body\gets\Call{EncodeMessage}{R,ctx,cert,T,prefix}$
  \State $evidence_i\gets\Call{ExecuteProtocolConnection}{e_i,body}$
  \Statex \hspace{\algorithmicindent}$\triangleright$ At the receiving XIR Adapter on $G_i$'s blockchain
  \State $accepted_i\gets\Call{Authenticate}{e_i,evidence_i}$
  \State \textbf{assert} $accepted_i$
  \State $C_i\gets\Call{MakeReceipt}{e_i,evidence_i,prefix,R,ctx}$
  \State \textbf{assert} $C_i.srcGateway=G_{i-1}\land C_i.dstGateway=G_i$
  \State \Call{RecordAcceptedEvidence}{$C_i,accepted_i$}
  \State $T\gets T\mathbin{\|}[C_i]$
  \State $prefix\gets\Call{ExtendPrefix}{prefix,C_i}$
  \If{$i<h$}
    \State \textbf{assert} \Call{ApprovedIngress}{$T,P.edges[i+1],\Phi_{i+1}$}
  \EndIf
\EndFor
\Statex $\triangleright$ At the destination XIR Gateway $G_h$
\State \textbf{assert} \Call{VerifySource}{$R,ctx,rid,mid,cert$}
\State \textbf{assert} \Call{VerifyOrderedReceipts}{$T,prefix,R,ctx$}
\State \textbf{assert} \Call{VerifyFinalBundle}{$T,G_h$}
\State \textbf{assert} $\neg consumed[mid]$
\State \Return \textsc{AtomicConsumeAndInvoke}(
\Statex \hspace{3em}$mid,R.destinationApp,m$)
\end{algorithmic}
\end{algorithm}

The ordered verification records and final commitment provide verification
continuity.  Stable operation identifiers, approved incoming XIR Adapters, and XIR
Adapters provide execution continuity.  Section~\ref{sec:design} defines these
operations in execution order.

%% file: sections/4-xir-design.tex
\section{XIR Design}
\label{sec:design}

XIR binds each protocol delivery to one application operation and its
source-certified XIR Registry version. XIR Gateways check profile activity and
approved XIR Adapters against the configuration available when they execute.  This section follows the execution order:
source initialization, protocol sends, XIR Gateway handoffs, and destination
execution.

Figure~\ref{fig:xir-protocol-design} shows the complete lifecycle.  The upper
part shows the protocol path.  The middle part shows how the application record
gains one verified delivery record at each blockchain.  The lower part shows
the checks performed by the destination.

\begin{figure*}[t]
\centering
\includegraphics[width=\textwidth]{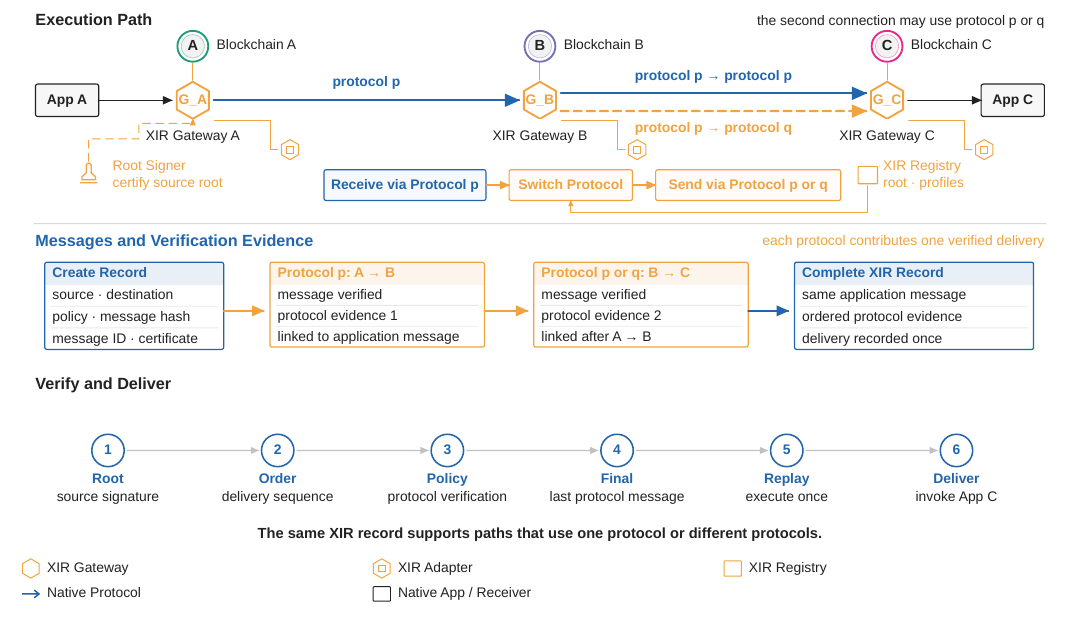}
\caption{\textbf{Execution and verification in XIR.} The same application record crosses each protocol connection. An XIR Gateway records the verification evidence produced by each protocol. The destination checks the source, delivery order, verification policy, final protocol message, and replay state before executing the application. Ring color marks the source, intermediate, and destination blockchain.}
\label{fig:xir-protocol-design}
\end{figure*}

\subsection{Route and Representation Initialization}

The source XIR Gateway authenticates the application and creates
\begin{equation}
\begin{aligned}
R=\langle{}&sourceGateway,sourceApp,destinationApp,\\
            &nonce,payloadHash\rangle,\\
ctx=\langle{}&requiredSecurity,policyHash\rangle .
\end{aligned}
\label{eq:xir-record}
\end{equation}
XIR Gateway and application identifiers are typed pairs $\langle kind,value\rangle$.
The XIR Gateway increments a per-caller nonce; the payload determines
\texttt{payloadHash}.

Two stable identifiers connect source authorization to destination execution.
The root identifier binds the record, context, and XIR Registry version; the
message identifier binds that root to the destination application.
Domain-separated hashes distinguish these uses:
\begin{align}
d_R&=H(\mathtt{XIR\_RECORD\_V1}\parallel\mathsf{Encode}(R)),\nonumber\\*
d_{ctx}&=H(\mathtt{XIR\_CONTEXT\_V1}\parallel\mathsf{Encode}(ctx)),\nonumber\\*
rid&=H(\mathtt{XIR\_RID\_V1}\parallel sourceGateway\nonumber\\*
   &\qquad\parallel d_R\parallel d_{ctx}\parallel registryVersion),\nonumber\\*
mid&=H(\mathtt{XIR\_MID\_V1}\parallel rid\parallel destinationApp).
\label{eq:xir-mid}
\end{align}
The root signer validates the finalized \texttt{RootCreated} event and signs
$rid$; the certificate $cert$ carries the corresponding XIR Registry version.

For each direct protocol connection, the XIR Registry resolves a verification profile
\begin{equation}
\begin{aligned}
\Phi_i=\langle{}&srcHash,dstHash,adapter,securityLevel,\\
                 &validAfter,validUntil,enabled\rangle .
\end{aligned}
\label{eq:xir-descriptor}
\end{equation}
The XIR Registry administrator assigns security levels on the policy's ordered scale.
A profile is eligible when its endpoint hashes match the edge, it is enabled
within its validity window, and
\begin{equation}
\Phi_i.securityLevel\ge ctx.requiredSecurity.
\label{eq:effective-policy}
\end{equation}
The XIR Router selects eligible protocol connections and XIR Adapters in path order; XIR Gateways
enforce the corresponding XIR Registry entries during execution.

\subsection{Cross-Chain Protocol Send and Evidence Acceptance}

Before protocol connection $i$, the outgoing XIR Adapter encodes the V-IR state for protocol
$p_i$.  The protocol endpoint sends this state over the configured direct
connection.  The incoming XIR Adapter receives the message through an
authenticated protocol callback, such as a Hyperlane Mailbox
callback~\cite{hyperlanemailbox} or a LayerZero endpoint
callback~\cite{layerzerooapp}.

The incoming XIR Adapter records enough information to associate the accepted
delivery with its endpoints, verifier profile, and application operation.
This information forms one hop receipt:
\begin{equation}
\begin{aligned}
C_i=\langle{}&srcGateway,dstGateway,profileHash,\\
              &evidenceHash,transitionHash,priorPrefix\rangle .
\end{aligned}
\label{eq:protocol-delivery-receipt}
\end{equation}
The protocol message digest or GUID supplies \texttt{evidenceHash};
\texttt{profileHash} identifies $\Phi_i$; and \texttt{transitionHash} binds
$d_R$, $d_{ctx}$, and adjacent XIR Gateways. The XIR Adapter records
\begin{equation}
\begin{aligned}
\kappa_i=H(&C_i.profileHash,C_i.evidenceHash,\\
      &C_i.transitionHash)
\end{aligned}
\label{eq:accepted-evidence}
\end{equation}
in \texttt{acceptedEvidence}.  This entry identifies the authenticated
protocol callback used by $C_i$.

\subsection{XIR Gateway Handoff and XIR Adapter Conversion}

At an intermediate XIR Gateway, the incoming XIR Adapter passes the accepted
receipt and accumulated V-IR state to the outgoing XIR Adapter.  For
$A\xrightarrow{p}B\xrightarrow{p}C$, both XIR Adapters integrate $p$. For
$A\xrightarrow{p}B\xrightarrow{q}C$, the inbound XIR Adapter integrates $p$ and
the outgoing XIR Adapter integrates $q$.  In both cases, $R$, $ctx$, $rid$, and
$mid$ remain unchanged.  The outgoing XIR Adapter changes only the protocol
encoding and endpoint used for the next connection.

An outgoing XIR Adapter must establish which incoming XIR Adapter accepted the
preceding evidence. The XIR Registry records the XIR Adapter that verifies each
profile, and the outgoing XIR Adapter enforces this binding through its approved
verifier mapping. The outgoing XIR Adapter receives a \texttt{ForwardRequest}
containing the V-IR state and its verifier, profile, evidence, and transition
arrays. For each preceding receipt $C_i$, the incoming XIR Adapter $a_{in}$ must satisfy
\begin{equation}
\begin{aligned}
&\mathsf{approvedPriorVerifiers}[C_i.profileHash]\\
&\qquad=a_{in},\\
&a_{in}.\mathsf{verify}(C_i.profileHash,\\
&\qquad C_i.evidenceHash,C_i.transitionHash)\\
&\qquad=\mathsf{true}.
\end{aligned}
\label{eq:approved-ingress}
\end{equation}
The next protocol send can extend the operation only from an incoming XIR
Adapter that is approved by the XIR Registry and recorded the preceding evidence.

The receipt chain binds each new delivery to the history already accepted.
Its initial commitment starts at
$q_0=H(\mathtt{XIR\_ROOT\_V1}\parallel rid)$ and advances as
\begin{equation}
\begin{aligned}
q_i=H(&\mathtt{XIR\_PREFIX\_V1}\parallel q_{i-1}\parallel\\
      &H(\mathtt{XIR\_HOP\_V1}\parallel\mathsf{Encode}(C_i))).
\end{aligned}
\label{eq:xir-prefix}
\end{equation}
The \texttt{priorPrefix} field in $C_i$ must equal $q_{i-1}$, and the
destination XIR Gateway in $C_i$ must equal the source XIR Gateway in $C_{i+1}$.
Every handoff preserves two invariants: the application identity remains
stable, and the accepted evidence extends exactly one ordered history.

The final delivery must authenticate the same receipt sequence that the
destination checks. This binding prevents accepted evidence from separate
executions from being combined into a new history. On the final hop, the
destination XIR Adapter authenticates the ordered evidence tuples through
\begin{equation}
b_h=\mathsf{BundleCommit}(C_1,\ldots,C_h).
\label{eq:xir-bundle}
\end{equation}
Here, $\mathsf{BundleCommit}$ binds the list length and each indexed
$(profileHash,evidenceHash,transitionHash)$ tuple in receipt order.
The prefix commits to receipt order, while $b_h$ binds that order to the final
protocol delivery.

\subsection{Destination Verification, Replay Protection, and Execution}

The destination XIR Gateway first recomputes the payload hash, $rid$, and $mid$ and
verifies the source certificate.  It scans the receipts in path order and, for
each $C_i$, checks endpoint continuity, \texttt{priorPrefix}, the transition
binding, profile activity, Eq.~\eqref{eq:effective-policy}, and the exact
accepted-evidence tuple in XIR Adapter state.  It then requires the final XIR Gateway
identity and the locally recomputed $b_h$ to match the values authenticated by
the destination XIR Adapter.

After these checks succeed, the XIR Gateway rejects an existing \texttt{consumed}
entry for $mid$, marks $mid$ as consumed, and invokes \texttt{xirReceive}.  If
the receiver call fails, the XIR Gateway reverts the transaction; EVM reversion
rolls back the associated state changes~\cite{eip140}.  A committed delivery
retains the consumed marker and one destination application effect.
These operations establish G1--G4 in the order defined in
Section~\ref{sec:model-assumptions}.

%% file: sections/5-analysis.tex
\section{Analysis}
\label{sec:analysis}

This section proves the destination properties defined in Section~\ref{sec:model},
analyzes how path information grows, and evaluates configuration and deployment
cost.

\subsection{Correctness under A1--A4}

An accepted delivery is a transaction in which the destination XIR Gateway
completes every check in Section~\ref{sec:design}, consumes $mid$, and commits
the receiver invocation.

\begin{lemma}[Prefix and bundle invariant]
\label{lem:trace-invariant}
After the destination verifies receipt $C_i$, the prefix commits exactly to
$[C_1,\ldots,C_i]$, adjacent receipt endpoints are continuous, and the bundle
accumulator commits to the evidence tuples of the same length.
\end{lemma}

\noindent\emph{Proof.}
For $i=0$, $q_0$ commits to the certified $rid$, and both the receipt list and
the evidence list are empty.  Assume the invariant holds after $C_{i-1}$.
Verification of $C_i$ checks $C_i.priorPrefix=q_{i-1}$, the source and
destination XIR Gateways, the registered profile and XIR Adapter, the transition hash,
and the accepted evidence tuple.  At a handoff, the XIR Adapter approval check
binds that tuple to the incoming XIR Adapter.  Equations~\eqref{eq:xir-prefix}
and~\eqref{eq:xir-bundle} then add the same receipt and evidence tuple to the
two commitments.  Collision resistance preserves their contents and order,
so the invariant holds after $C_i$.  Induction gives the result for all $h$
receipts.

\begin{theorem}[Conditional XIR delivery]
\label{thm:end-to-end}
If an XIR destination commits delivery under A1--A4, then G1 source
authorization, G2 ordered evidence integrity, G3 per-hop policy enforcement,
and G4 at-most-once delivery hold.
\end{theorem}

\noindent\emph{Proof.}
Root verification recomputes $d_R$, $d_{ctx}$, $rid$, and $mid$ from the
accepted record and validates the signature of the registered root signer.  Under A2--A4,
these checks establish G1. If the source block is rolled back before root
certification, no certified root matches that event, and the destination
rejects it. A source rollback after certification is excluded by the source
condition of A2: the root signer's waiting rule must cover source finality.
By Lemma~\ref{lem:trace-invariant}, the prefix and
bundle commitments contain one endpoint-continuous sequence of authenticated
protocol deliveries.  The destination also verifies the final XIR Adapter and
bundle, establishing G2 under A1--A4. A1 applies to every hop, including
those that start on an intermediate blockchain. If an intermediate block is
removed after downstream acceptance, A1 fails for that hop; each protocol's
confirmation setting must cover the finality of the blockchain it reads from.
Each receipt resolves to an enabled
profile in its validity window, names the receipt endpoints, and satisfies
Eq.~\eqref{eq:effective-policy}; these checks establish G3 under A3.  Finally,
the XIR Gateway rejects an existing \texttt{consumed[mid]} entry, sets the entry,
and invokes the receiver in one transaction.  A failed invocation reverts the
state update under EVM semantics~\cite{eip140}; after a committed invocation,
the consumed entry rejects every later use of $mid$.  These steps establish G4 under A2.
If the destination blockchain rolls back the committing block, the consumed entry and the application effect are removed together.
A successful redelivery after rollback commits one effect in the final
destination state, so the message identifier still corresponds to at most one
effect. The identifier is fixed when the source record is created; if an
intermediate rollback causes a hop to be sent again, the destination rejects
a second copy as a duplicate.

A1 defines accepted protocol evidence, A2 covers XIR Gateway and XIR Adapter execution,
A3 authenticates profiles and ingress, and A4 binds identifiers, prefixes,
signatures, and bundles. Section~\ref{sec:experiment-execution} exercises 780
malformed or replayed executions against these checks.

\subsection{Cost Analysis}
\label{sec:execution-analysis}
\label{sec:empirical-cost-analysis}

\subsubsection{Path Information and Verification Work}

Let $h$ be the number of cross-chain protocol connections in a path $P$,
$s=c(P)$ the number of protocol switches, and $k$ the number of earlier
deliveries recorded at one handoff.  By definition, $0\le s\le h-1$.
A \texttt{ForwardRequest} carries variable-length arrays of fixed-size entries
for the prior verifier, profile, evidence, and transition values.
Let $B_0$ be the fixed message size in bytes and $B_C$ the bytes per prior
receipt. Verification work $W$ assumes a fixed amount of work per receipt.
The information and verification work at one handoff are
\begin{equation}
\begin{aligned}
B_{\mathrm{handoff}}(k)&=B_0+kB_C=\Theta(k),\\
W_{\mathrm{handoff}}(k)&=\Theta(k).
\end{aligned}
\label{eq:handoff-growth}
\end{equation}
The final receipt list $T_h$ contains $h$ receipts, and the destination scans them once:
\begin{equation}
|T_h|=h,\qquad W_{\mathrm{destination}}(h)=\Theta(h).
\label{eq:destination-growth}
\end{equation}

The current encoder retransmits the complete path history on every later
connection.  Total XIR information carried across all connections is
\begin{equation}
B_{\mathrm{route}}(h)=
\sum_{i=1}^{h}\bigl(B_0+(i-1)B_C\bigr)=O(h^2).
\label{eq:route-growth}
\end{equation}
Experiment 2 measures the local cost of one more earlier delivery at a fixed
final handoff.  Experiment 3 measures the contribution of a protocol switch to
the complete path while controlling for $h$ and the starting protocol.
Equations~\eqref{eq:destination-growth} and~\eqref{eq:route-growth} separate
destination verification from repeated transmission of the path history.

\subsubsection{Reachability Provisioning}

Let $\mathcal R^*$ be a target set of reachable ordered blockchain pairs.  A
direct-connectivity baseline creates one point-to-point protocol configuration
for each pair in $\mathcal R^*$.  Full directed coverage over $N$ blockchains
requires
\begin{equation}
D_{\mathrm{pair}}=N(N-1)=O(N^2).
\label{eq:pair-provisioning}
\end{equation}
For $N=286$, this expression gives 81,510 point-to-point protocol configurations.
The cross-protocol closure in Eq.~\eqref{eq:observed-reachability-results}
contains 78,953 ordered pairs (96.86\% of full directed coverage). Matching
that target with direct connections requires 78,953 configurations.  Of these,
67,018 establish pairs that are not directly connected in the measured graph.
XIR forms these paths from the existing connections and adds no direct
point-to-point configuration between the newly reachable blockchain pairs.

The XIR architecture has two separate deployment units.  It places one XIR
Gateway on each participating blockchain and adds one integration for every
protocol attached to that blockchain:
\begin{equation}
G_{\mathrm{XIR}}=N,\qquad
I_{\mathrm{XIR}}=\sum_{p\in\mathcal P}|V_p|.
\label{eq:xir-provisioning}
\end{equation}
For a fixed protocol set, $I_{\mathrm{XIR}}\le|\mathcal P|N$.  XIR Gateway
deployments and protocol integrations therefore each grow as $O(N)$.  In the
measured graph, $G_{\mathrm{XIR}}=286$ and $I_{\mathrm{XIR}}=493$.  These
components combine 15,965 existing protocol-labeled directed connections into
78,953 reachable ordered pairs.  Every path uses XIR Gateways, XIR
Adapters, accepted application paths, and correctly configured protocol
connections.

\subsubsection{Measured Deployment Units}
\label{sec:deployment-units}

Deployment gas provides one common cost unit for the measured contracts.
Figure~\ref{fig:xir-analysis-overview} organizes the deployment counts and
measured per-blockchain costs. The calculations below distinguish XIR core
contracts, protocol stacks, XIR Adapters, and endpoint configuration.  One \texttt{XIRGateway} and one
\texttt{XIRRegistry} consume 2,134,905 gas.  Across 286 blockchains, the XIR
core total is 610,582,830 gas. The measured Hyperlane~\cite{hyperlanemailbox}
and LayerZero~\cite{layerzero} protocol stacks consume 5,529,809 and 27,881,442 gas per blockchain.  Their observed
coverage gaps are 118 and 144 blockchains, giving modeled expansion totals of
652,517,462 and 4,014,927,648 gas.  These extrapolations describe the core and protocol-stack deployment
components under their respective coverage assumptions. A complete
reachability deployment also includes XIR Adapters and endpoint configuration.
A coverage gap is a blockchain on which the dataset records no events for that protocol.

\begin{figure*}[t]
\centering
\includegraphics[width=\textwidth]{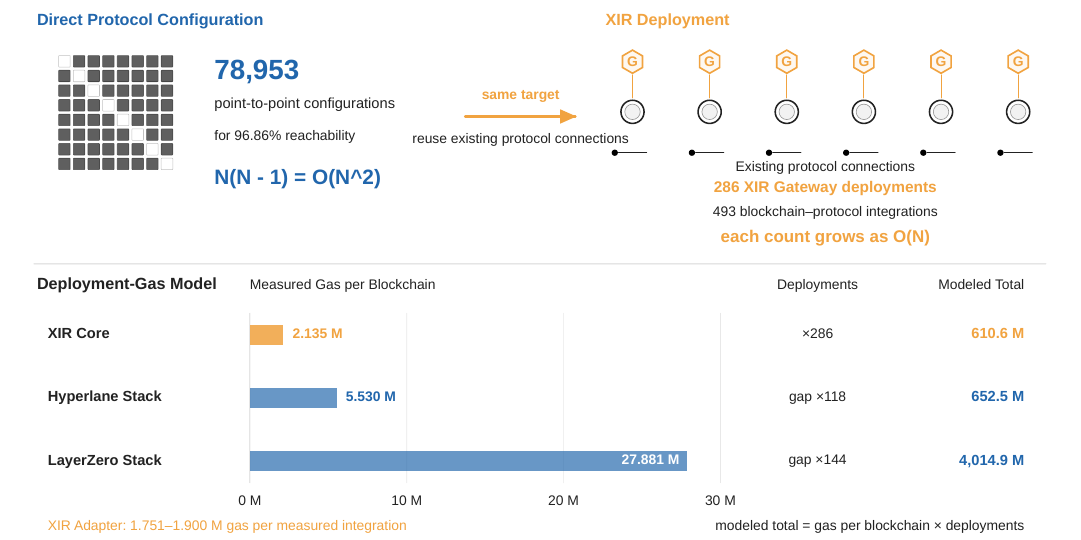}
\caption{\textbf{Configuration and deployment cost.} The upper part reports point-to-point protocol configurations, XIR Gateway deployments, and protocol integrations as separate units. The lower part compares deployment cost in gas by multiplying each measured per-blockchain cost by its deployment count.}
\label{fig:xir-analysis-overview}
\end{figure*}

Measured XIR Adapters for Hyperlane and LayerZero add 1,750,897 and 1,899,585 gas,
respectively, outside the 286-blockchain XIR core total.  In the configuration
benchmark, 100 first writes to unique remote identifiers yield median costs of 95,894 gas for
\texttt{Router.enrollRemoteRouter}~\cite{hyperlanerouter} and 47,387 gas for
\texttt{OAppCore.setPeer}~\cite{layerzerooapp}. The local path-isolation matrix records 125,583,834 gas across 79
contract deployments and 137,989,270 gas across 264 deployment and
configuration transactions.  The measurements use three- and five-blockchain
deployments with protocol-specific XIR Adapters; the 286-blockchain totals apply these
units to Eq.~\eqref{eq:xir-provisioning}. A blockchain--protocol integration
denotes an architectural placement. Each measured XIR Adapter binds one immutable
remote domain or endpoint identifier, so a concrete deployment counts the
XIR Adapter instances and remote bindings that its paths need. Shared contracts
are counted once.

\subsubsection{Direct Connections and Reused Paths}
\label{sec:cost-crossover}

A direct connection trades an initial deployment or configuration cost for a
lower cost per execution. Four cases determine the initial cost of the direct
alternative (Figure~\ref{fig:direct-vs-xir-cases}). First, if a direct application connection already works, the
direct alternative needs no new work. Second, if both endpoints already
support the same protocol, a new application connection needs remote bindings
and any application-specific security configuration. Third, if the source
supports only Hyperlane and the destination only LayerZero, XIR can reuse an
existing Hyperlane--LayerZero path through an intermediate blockchain, whereas
a direct alternative must add Hyperlane at the destination or LayerZero at the
source, together with the required application configuration and services.
Fourth, if the application can add neither missing protocol, for example
because the protocol operator has not deployed it on that blockchain, the
direct alternative is unavailable at any execution volume.

Let $D_X$ and $D_D$ denote the new deployment and configuration costs of XIR
and a direct alternative, expressed in the same unit. Let $E_X$ and $E_D$
be their costs per completed execution. The total costs for $n$ executions are
\begin{equation}
 C_X(n)=D_X+nE_X,\qquad C_D(n)=D_D+nE_D.
 \label{eq:total-execution-cost}
\end{equation}
When $D_D>D_X$ and $E_X>E_D$, XIR costs less for
$n<(D_D-D_X)/(E_X-E_D)$; direct delivery becomes cheaper beyond this
crossover. If XIR has both higher initial and higher recurring cost, there
is no positive crossover. First-use storage writes can be modeled as a
separate term for the first execution, so the model charges them once.

Two deployment perspectives answer different decisions. Building XIR for a
new path includes its XIR Gateways, XIR Registries, XIR Adapters, and all necessary
bindings. Reusing an already configured XIR path charges only new work.

XIR Gateways and XIR Registries serve every application on their blockchain, so
a deployment pays the cost of building the first path once and later paths reuse it.
Section~\ref{sec:cost-practical} measures the execution and deployment
components for the Hyperlane and LayerZero examples.

%% file: sections/6-implementation-evaluation.tex
\section{Implementation and Evaluation}
\label{sec:implementation}

The evaluation addresses three questions.  \textbf{Q1:} Can XIR deliver
messages over same-protocol and cross-protocol paths while enforcing its
verification checks?  \textbf{Q2:} How do earlier same-protocol deliveries
change the cost of one fixed final protocol switch?  \textbf{Q3:} How do path
length and the number of protocol switches affect total execution cost?

\subsection{Prototype and Experimental Setup}

The prototype includes Solidity XIR Gateways, an XIR Registry,
XIR Adapters for Hyperlane~\cite{hyperlanemailbox} and LayerZero~\cite{layerzerooapp}, destination applications, and a stateful Python
orchestrator.  The implementation computes the same hashes in Solidity, Python, and Rust and
connects to Hyperlane Mailbox callbacks~\cite{hyperlanemailbox} and
LayerZero endpoint callbacks~\cite{layerzerooapp}.  The artifact
contains contract sources, deployment configurations, protocol-message records,
generated aggregates, and analysis scripts.

The controlled environment has three or five blockchains running
Besu QBFT~\cite{besuqbft}, with four validators each. Every blockchain runs XIR, Hyperlane, and
LayerZero stacks. The Hyperlane stack contains a Mailbox, MerkleTreeHook,
MessageIdMultisigIsm, validator, and relayer; the LayerZero stack contains
EndpointV2, SendUln302, ReceiveUln302, DVN, Executor, fee library, price feed,
and treasury. The public-testnet executions use OP Sepolia, Arbitrum Sepolia,
Base Sepolia, and Solana Devnet.
Figure~\ref{fig:experiment-environment} maps the experiments to these
environments.

\begin{figure*}[t]
\centering
\includegraphics[width=\textwidth]{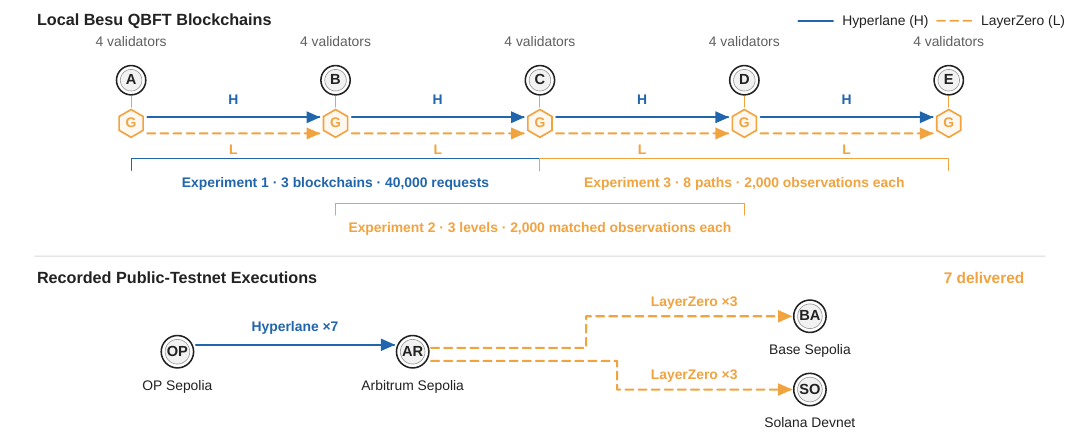}
\caption{\textbf{Experimental environments.} Experiment 1 uses three local blockchains for same-protocol and cross-protocol paths. Experiments 2 and 3 use five local blockchains. The seven successful public-testnet executions use OP Sepolia, Arbitrum Sepolia, Base Sepolia, and Solana Devnet.}
\label{fig:experiment-environment}
\end{figure*}

The toolchain uses Solidity 0.8.28, Foundry 1.5.1, Python 3.13.11, Node.js
24.11.1, Hyperlane Core 11.3.1, Hyperlane SDK 36.0.0, and LayerZero V2 3.0.168.
The workloads fix the path, payload schedule, and concurrency limit.  Each
successful request must have a unique nonce, records for every protocol
message, and exactly one destination effect.
The Experiment 2/3 campaign contains 22,000 completed executions: 2,000 for
each of 11 protocol paths.  Experiment 2 uses three matched path groups, and
Experiment 3 uses eight principal paths.
The measured graph has directed diameter four among finite reachable pairs
(Section~\ref{sec:hop-count}); Experiment 3 covers all four path lengths, while
Experiment 2 varies the number of earlier deliveries before a fixed final
switch.

Table~\ref{tab:experimental-settings} summarizes the controlled workloads.
H denotes Hyperlane and L denotes LayerZero. In Experiment 1, HH and LL use
native same-protocol forwarding; HL and LH use XIR for the protocol handoff.
Protocol verification roles are separate from the four QBFT consensus
validators on each blockchain.

\begin{table*}[t]
\caption{\textbf{Experimental settings and workloads.} Settings are grouped by the experiment to which they apply. The Experiment 1 delivery settings follow its frozen profile; the adversarial workload is evaluated separately.}
\label{tab:experimental-settings}
\centering
\small
\setlength{\tabcolsep}{5pt}
\renewcommand{\arraystretch}{1.12}
\begin{tabular}{@{}p{0.17\textwidth}p{0.23\textwidth}p{0.55\textwidth}@{}}
\toprule
Scope & Setting & Value \\
\midrule
Local experiments & Consensus and protocols & Besu QBFT; four consensus validators per blockchain; Hyperlane and LayerZero \\
\midrule
Experiment 1\newline Delivery & Blockchains and paths & Three blockchains; native HH and LL; XIR cross-protocol HL and LH \\
 & Requests and payload & 10,000 requests per path; application payloads of 32, 64, 96, and 128 bytes \\
 & Warm-up and retries & Zero warm-up requests; retries excluded from the designated request denominator \\
 & Hyperlane verification & MessageIdMultisigIsm; one validator per origin, threshold one; one relayer; reorg period zero \\
 & LayerZero verification & One required DVN, zero optional DVNs, one Executor; one source confirmation; self-hosted research worker \\
\midrule
Experiment 2 & Blockchains and paths & Five blockchains; HL, HHL, HHHL; one to three earlier deliveries before the final switch \\
 & Matched workload & 2,000 requests per path, matched by route sequence \\
 & Latency equivalence & Margin of $\pm0.5$ seconds per earlier delivery; 90\% confidence intervals \\
\midrule
Experiment 3 & Blockchains and paths & Five blockchains; H, HH, HHH, HHHH, HL, HLH, HLHL, LHLH \\
 & Workload and factors & 2,000 requests per path; one to four protocol connections, zero to three protocol switches \\
 & Cost-model intervals & 95\% confidence intervals for regression coefficients \\
\midrule
Adversarial workload & Cases and repetitions & 13 classes; HL and LH; 30 repetitions per class and direction \\
\bottomrule
\end{tabular}
\end{table*}

Figure~\ref{fig:evaluation-results} summarizes the cost and latency measurements in
Experiments 2 and 3.  Experiment 1 is reported directly in
Section~\ref{sec:experiment-execution}.

\begin{figure*}[t]
\centering
\makebox[\textwidth][c]{\includegraphics[width=1.08\textwidth]{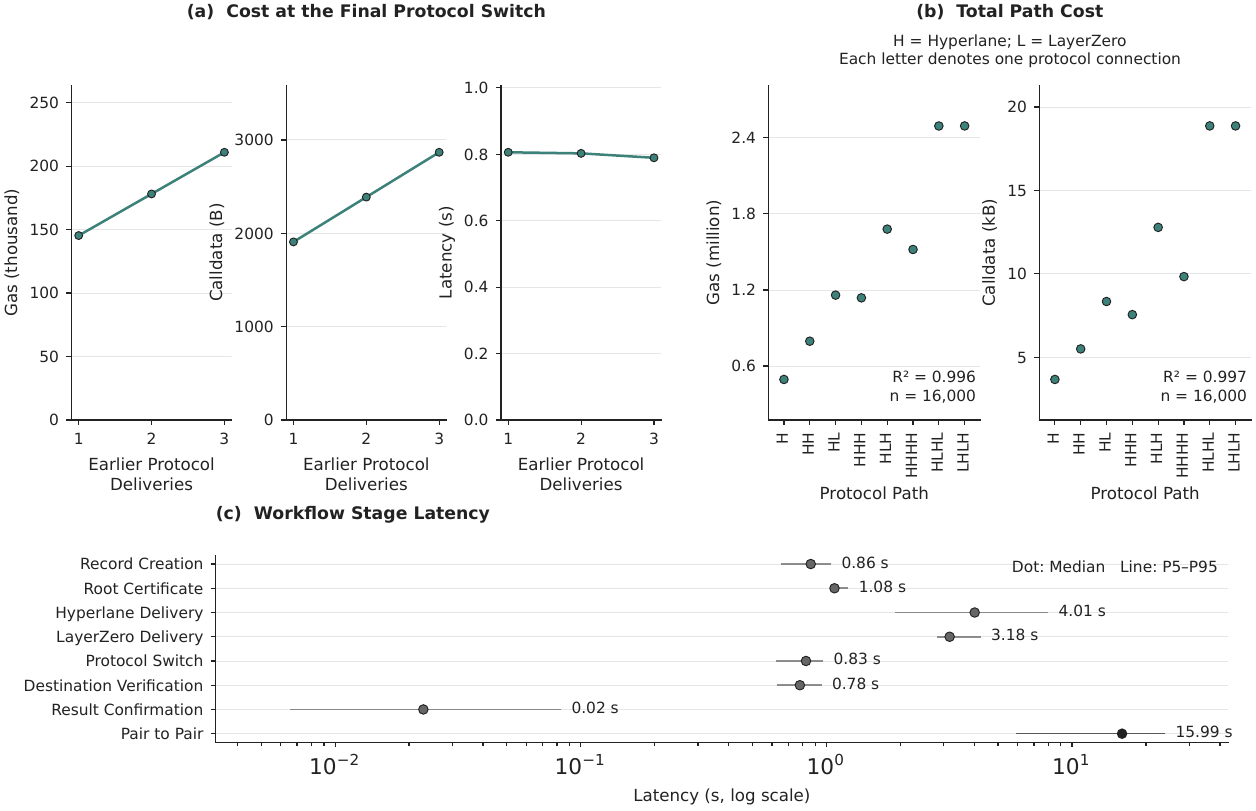}}
\caption{\textbf{Execution cost and workflow latency.} (a) Experiment 2 uses 6,000 matched observations, with 2,000 at each number of earlier deliveries. (b) Experiment 3 reports mean costs for eight paths across 16,000 executions; each letter denotes one protocol connection (H: Hyperlane; L: LayerZero). The displayed $R^2$ values describe the additive models. Teal denotes aggregate measurements in (a) and (b). (c) Stage latencies use all recorded instances from the 22,000-execution campaign; blue denotes Hyperlane delivery and orange denotes LayerZero delivery. Points show medians; lines show the 5th--95th percentiles.}
\label{fig:evaluation-results}
\end{figure*}

\subsection{Experiment 1: Pair-to-Pair Feasibility}
\label{sec:experiment-execution}

Experiment 1 varies the protocols used by a two-hop path and measures delivery,
single destination effects, receipt order, and rejection checks.  The
three-blockchain environment executes two same-protocol paths, HH
(Hyperlane--Hyperlane) and LL (LayerZero--LayerZero), and two cross-protocol
paths, HL and LH.  Each path carries 10,000 requests under one schedule.  A
conformance workload
runs 13 adversarial classes in both cross-protocol directions, with 30
repetitions per class and direction.

HH and LL create no XIR record and skip XIR Gateway verification; in Experiments 2 and 3, every path, including HH, carries XIR records.

All 40,000 local requests complete with exactly one destination effect.  Every
cross-protocol request records the verification result of both protocol
deliveries and the protocol switch.  Across 20,000 records, the source XIR Gateway step
averages 49,370.9 gas and 596 calldata bytes; across 20,000 records, the
protocol-switch step averages 136,508.7 gas and 1,812 bytes.  The 780 adversarial case-runs exhibit the expected rejection behavior for
payload, context, profile, trace, evidence, bundle, approved-ingress, and
replay checks. Rejected submissions leave consumption and application state
unchanged; replay cases retain the initial application effect and reject
duplicate execution. The public-testnet ledger contains eleven attempts:
seven delivered and four encountered protocol failures. The seven deliveries
comprise one Hyperlane EVM path, three Hyperlane-to-LayerZero EVM paths,
and three Hyperlane-to-LayerZero EVM-to-Solana paths.
Table~\ref{tab:public-testnet-transactions} gives the destination transaction or
signature for every successful execution.
Together, the local delivery workload, adversarial cases, and selected
public-testnet paths establish execution feasibility and conformance to the
implemented verification checks.

\begin{table*}[t]
\caption{\textbf{Successful public-testnet executions.} Each link resolves to the complete destination transaction or Solana signature.}
\label{tab:public-testnet-transactions}
\centering
\footnotesize
\setlength{\tabcolsep}{4pt}
\begin{tabular}{@{}lllll@{}}
\toprule
Blockchain path & Protocol path & Terminal network & Outcome & Transaction or signature \\
\midrule
OP Sepolia $\rightarrow$ Arbitrum Sepolia
  & Hyperlane & Arbitrum Sepolia
  & Delivered
  & \texttt{\href{https://sepolia.arbiscan.io/tx/0x8ba2a59cea9868fe7171f6b59d9a380e18563bfc1c07f05cbfa27860fe680497}{0x8ba2\ldots{}0497}} \\
OP Sepolia $\rightarrow$ Arbitrum Sepolia $\rightarrow$ Base Sepolia
  & Hyperlane $\rightarrow$ LayerZero & Base Sepolia
  & Delivered
  & \texttt{\href{https://sepolia.basescan.org/tx/0x23e7fe354df890d760879ae5c525ac831a5fdc9be76f8670e335fc7633a0b1b5}{0x23e7\ldots{}b1b5}} \\
OP Sepolia $\rightarrow$ Arbitrum Sepolia $\rightarrow$ Base Sepolia
  & Hyperlane $\rightarrow$ LayerZero & Base Sepolia
  & Delivered
  & \texttt{\href{https://sepolia.basescan.org/tx/0x1e86f4172e7fe62a18c4fdcbcc9d4fac1b45adb91185f60300c4f1cc11042074}{0x1e86\ldots{}2074}} \\
OP Sepolia $\rightarrow$ Arbitrum Sepolia $\rightarrow$ Base Sepolia
  & Hyperlane $\rightarrow$ LayerZero & Base Sepolia
  & Delivered
  & \texttt{\href{https://sepolia.basescan.org/tx/0xd54cf92e414784251eead8f2775dd828ab11019008dfb2b3fe2399d64214736e}{0xd54c\ldots{}736e}} \\
OP Sepolia $\rightarrow$ Arbitrum Sepolia $\rightarrow$ Solana Devnet
  & Hyperlane $\rightarrow$ LayerZero & Solana Devnet
  & Delivered
  & \texttt{\href{https://solscan.io/tx/2Ed9h2KdJeq81XATnhhBSGie75tEdTb43v4frjChv65nw19vWmtE4SaT2J4vPSvNAHFTA4H7941h3rP5ZC6j5KBc?cluster=devnet}{2Ed9\ldots{}j5KBc}} \\
OP Sepolia $\rightarrow$ Arbitrum Sepolia $\rightarrow$ Solana Devnet
  & Hyperlane $\rightarrow$ LayerZero & Solana Devnet
  & Delivered
  & \texttt{\href{https://solscan.io/tx/5bUVz9PvqUxLjVvLeNHaUzBNU4cwLGoxjh8QNocNyQwzXTY1vkxtNiESodnoQVv6UaSmmkU8EMrGgnegnxNUBeUQ?cluster=devnet}{5bUV\ldots{}UBeUQ}} \\
OP Sepolia $\rightarrow$ Arbitrum Sepolia $\rightarrow$ Solana Devnet
  & Hyperlane $\rightarrow$ LayerZero & Solana Devnet
  & Delivered
  & \texttt{\href{https://solscan.io/tx/58JFKrmYw9umAzAEpEVkqAuQmVRxRVoxKwCcaXFj82bwLBonYqBGUvdhhbKuLWGm9i4apWuyQzZAeWwGFTmujvu2?cluster=devnet}{58JF\ldots{}mujvu2}} \\
\bottomrule
\end{tabular}
\end{table*}

\subsection{Experiment 2: Local Cost at the Final Protocol Switch}
\label{sec:experiment-transition}

Experiment 2 measures the local cost of earlier deliveries at one fixed final
protocol switch.  The independent variable is the number of earlier
same-protocol deliveries, $k\in\{1,2,3\}$.  The corresponding paths are HL,
HHL, and HHHL.  Each delivery adds one verification
record to the XIR message.  The dependent variables are gas, calldata, and handoff latency.  The
incoming XIR Adapter lookup, evidence query, bundle update, and outgoing protocol
send remain fixed.  The matched sample contains
6,000 handoff observations, with 2,000 observations for each number of earlier
deliveries.

One more earlier delivery adds exactly 480 calldata bytes and approximately
32,791 gas (90\% CI: 32,791--32,792) to the final handoff.  The
handoff-latency slope is $-0.0083$ seconds per earlier delivery, with a 90\% interval of
$[-0.0122,-0.0043]$ seconds; it lies within the prespecified $\pm0.5$-second
equivalence margin.  Each earlier delivery therefore adds a fixed amount of history and
verification work at the final handoff. Across the three measured history
lengths, gas and calldata grow linearly, while the latency slope remains
within the prespecified equivalence margin.

\subsection{Experiment 3: Pair-to-Pair Cost of Protocol Switches}
\label{sec:experiment-cost-growth}

Experiment 3 measures how each protocol switch contributes to total path cost.  The
five-blockchain deployment evaluates H, HH, HHH, HHHH, HL, HLH, HLHL, and
LHLH.  These eight paths span $h\in\{1,2,3,4\}$ hops and
$s\in\{0,1,2,3\}$ protocol switches.  Each path has 2,000 executions, for
16,000 observations in total.  For each path-level dependent variable
$Y$, the additive model is
\begin{equation}
Y=\beta_0+\beta_h h+\beta_s s+\beta_L\mathbb{1}[\text{starts with L}]+\epsilon .
\label{eq:experiment-three-model}
\end{equation}
The predictors are $h$, $s$, and the starting protocol.

One additional protocol connection contributes 344,658 gas (95\% CI:
344,548--344,767) and 2,130 calldata bytes.  For paths of the same length and
starting protocol, one additional protocol switch contributes 312,075 gas
(95\% CI: 311,973--312,177) and 2,887 calldata bytes.  The gas and calldata
models explain the observed variation with $R^2=0.996$ and $R^2=0.997$,
respectively.
This coefficient includes coordination and verification across the complete
path.  Experiment 2 instead measures one more earlier delivery at a fixed
handoff.  Across the measured one-to-four-hop paths, total gas and calldata are
approximately additive in path length and protocol-switch count. The complete-history
encoding still incurs the quadratic aggregate transmission term derived in
Eq.~\eqref{eq:route-growth}; the fitted coefficients describe the evaluated
path lengths and protocol combinations.

Figure~\ref{fig:evaluation-results}(c) separates cross-chain protocol delivery
from XIR processing. A Hyperlane delivery has a median latency of 4.015
seconds, and a LayerZero delivery has a median latency of 3.175 seconds. The
XIR protocol-switch and destination-verification stages have median latencies
of 0.826 and 0.780 seconds. The complete path has a median latency of 15.987
seconds. Section~\ref{sec:cost-practical} attributes the measured gas to XIR
and native protocol contracts and converts it into fees.

\input{sections/6-cost-revision}

%% file: sections/6-cost-revision.tex
\subsection{Cost Attribution, Direct Alternatives, and Fees}
\label{sec:cost-practical}

Experiment 3 reports the total gas of each path. This subsection answers
three follow-up questions about that cost: which contracts consume the gas,
when a direct connection is cheaper, and what the gas costs at public-network
prices. The three parts below answer them in turn.

\subsubsection{Execution Stages and XIR Attribution}
\label{sec:cost-attribution}

Each execution runs as a sequence of transactions. In the original 22,000
executions, these transactions fall into five non-overlapping stages, named
after the steps in Figures~\ref{fig:xir-internal-workflow}
and~\ref{fig:xir-protocol-design}: Create Record, Send via Protocol, Receive
via Protocol, Switch Protocol, and Verify and Deliver. Create Record writes
the record and emits the source root that the Root item of
Figure~\ref{fig:xir-protocol-design} signs. Switch Protocol runs at the XIR
Gateway. Verify and Deliver comprises the Root, Order, Policy, Final, and
Replay checks, followed by the Deliver step. The root signer certifies the
source root off-chain after Create Record; certification adds no transaction
to the plotted cost.
A receive transaction also runs the authenticated callback
into the XIR Adapter, so its gas belongs partly to the protocol and partly to
XIR.

The switch transaction performs the protocol switch itself. It uses
145,292 gas in HL, 178,073 gas in HHL, and 210,874 gas in HHHL, because each
earlier delivery adds to the history that the switch verifies.
The 312,075-gas switch coefficient of Experiment 3 is larger because it also
includes the protocol hop after the switch, and in this deployment a
LayerZero hop uses more gas than a Hyperlane hop. For example, HL uses
363,302 more gas than HH. The switch accounts for 145,292 gas of this
difference, and the send and receive of the LayerZero hop account
for another 217,795 gas.
Experiment 1 measured its 136,508.7-gas switch in a separate three-blockchain deployment (Section~\ref{sec:experiment-execution}).

To separate XIR work from protocol work inside each transaction, a
supplementary sample records opcode traces, which list every contract call
that a transaction executes together with its gas. The sample contains 32
executions for each of the 11 paths and 32 direct executions through each
protocol, 416 completed executions in total. The sample runs on fresh
five-blockchain QBFT state with the frozen native protocol and XIR contract
code and the same verifier thresholds and confirmation settings as the
original campaign. The direct baselines are thin native applications that
authenticate the native callback and invoke the same destination receiver
code; they carry no XIR record or root certificate.

Each contract is charged its exclusive gas, which is the gas it spends itself
after subtracting the calls it makes to other contracts. Transaction intrinsic
gas, which includes calldata, and refunds are reported separately. This
accounting covers all 3,808 sample transactions, and every decomposition sums
to the transaction receipt's \texttt{gasUsed}. The XIR column counts
execution in XIR contracts. Extra protocol work and calldata caused by
carrying the history appear in their own columns, so the difference from
direct delivery gives the total incremental cost of XIR.

For a one-hop Hyperlane path, XIR adds 241,840 gas to direct delivery
(494,848 versus 253,008 gas). XIR contracts execute 195,717 gas, and the
larger message adds 59,443 gas of calldata and 1,499 gas in protocol
contracts; removing the thin native
application saves 14,819 gas. On longer paths, XIR contracts account for
38.6\% of the gas of HL and 45.6\% of HLHL. Native protocol execution,
calldata, and the destination receiver account for the rest.

Figure~\ref{fig:cost-stage-cube} combines the per-transaction and per-path
views: each position identifies a path and transaction stage, and each stack
separates exclusive gas by contract. The stage names match Figures~\ref{fig:xir-internal-workflow} and~\ref{fig:xir-protocol-design}.
The figure groups transactions into 13 stages, with absent stages left blank;
a single path uses at most 17 transactions.
HLHL uses 17 transactions in execution order. HLHL switches protocol three times, so
it contains every kind of transaction in the experiments. XIR gas is zero in the two LayerZero decentralized verifier network
(DVN) verifications. Each commit stage averages 97 XIR gas: the endpoint
calls the XIR Adapter's \texttt{allowInitializePath} view function when the path
is first initialized; later commits need no such call. In every other transaction, XIR contracts do part of the work. At each
send, the XIR Adapter checks that every earlier delivery was accepted and
packs the history into the outgoing message. At each receive step, the XIR Adapter
callback records the evidence of this hop and of every earlier hop. Each
switch and the final delivery run the XIR Gateway and XIR Registry to
verify the history. The same kind of transaction uses more XIR gas later in
the path, because it carries and checks a longer history. A Hyperlane
delivery uses 56,708 XIR gas at hop~1 and 109,746 at hop~3, a LayerZero
executor delivery uses 82,300 at hop~2 and 135,264 at hop~4, and
Switches 1--3 in the HLHL sample use 111,012, 139,442, and 167,728.

The supplementary tables list every hop and contract, and
original-22000-per-hop-gas.csv gives the stage means of the original 22,000
executions. Figure~\ref{fig:cost-stage-cube} sums LayerZero DVN verification, commit,
and executor delivery into one receive stage. Verify and Deliver executes in one transaction. At Switch 1 in the HL sample, XIR contracts execute 110,001 gas.
The XIR Gateway verifies the incoming history with 48,005 gas, and the XIR
Registry returns the registered root and profiles for this check with 16,740
gas. The transition recorder logs each switch for experimental bookkeeping;
the destination's G1--G4 checks do not read its state. It uses 40,179 gas, and the XIR Adapter uses the remaining 5,077 gas.
XIR gas for Verify and Deliver grows with path length, from 98,670 gas in H
to 125,726 gas in HL and 179,836 gas in HLHL, because the destination checks
every earlier delivery in the record.

\begin{figure*}[t]
\centering
\includegraphics[width=\textwidth]{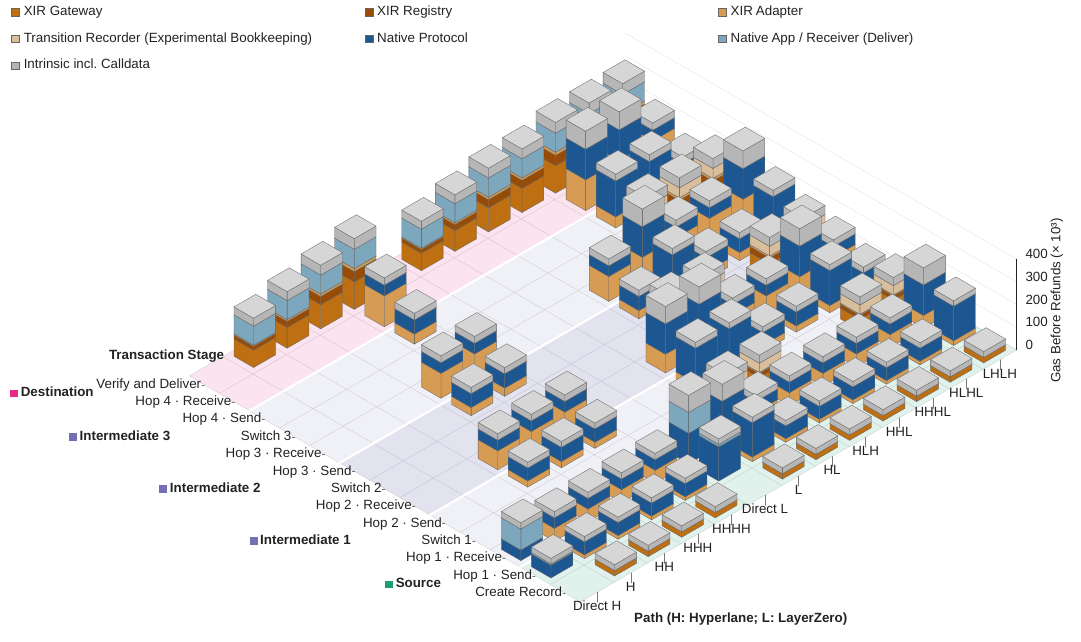}
\caption{\textbf{Where each path spends gas, by transaction stage and contract.}
Each stack is a mean of 32 executions and shows exclusive contract gas and
intrinsic gas before refunds. A blank position is a stage the path does not
execute. The 13 stages are Create Record, send and receive for Hops 1--4,
Switches 1--3, and Verify and Deliver. A LayerZero receive comprises three
transactions, summed here; the supplementary tables list them separately.
Destination checks and delivery execute in one transaction. The transition
recorder is experimental bookkeeping; destination checks do not read its
state. Floor bands mark the blockchain that executes each stage: source,
intermediate, and destination. Receipt means after refunds, in path order,
Direct H to LHLH, are 253.0, 494.8, 794.4, 1142.6, 1523.4, 471.7, 712.7,
1159.1, 1683.2, 1533.7, 1954.2, 2495.9, 2505.0 thousand gas. The
per-transaction tables behind this figure
(\texttt{sample-per-hop-attribution.csv},
\texttt{sample-xir-contract-stage.csv}, \texttt{sample-xir-by-stage.csv},
\texttt{sample-attribution.csv}) are in the supplementary tables and in the
XIR repository
(\href{https://github.com/lysrain21/XIR}{github.com/lysrain21/XIR},
directory \texttt{experiments/cost-revision-20261007/analysis}).
H: Hyperlane; L: LayerZero.}
\label{fig:cost-stage-cube}
\end{figure*}

\subsubsection{When a Direct Connection Is Preferable}
\label{sec:direct-preferable}

Figure~\ref{fig:direct-vs-xir-cases} shows the four cases;
Figure~\ref{fig:cost-crossover} plots $C(n)$ for case (c).
The answer depends on which case of Section~\ref{sec:cost-crossover}
applies. When both endpoints share a protocol, a direct connection is cheaper
from the first execution. The direct Hyperlane and LayerZero baselines consume
253.0 and 471.7 thousand gas per execution, compared with 1,159.1 thousand for
HL. If native application contracts already exist, their measured
remote-binding and option configuration costs are 91,266 and 180,026 gas for
H and L. These totals cover both endpoints: two remote bindings for H,
and two peer bindings plus two option writes for L. They use thin native
applications; Section~\ref{sec:deployment-units} reports single-call medians
for the upstream Router/OApp contracts. Each of these one-time costs is smaller than the extra gas of a
single HL execution.

When the source has only H and the destination only L, a direct alternative
adds a native protocol stack at the missing endpoint, together with
communication application contracts and their configuration. The private
deployment measurements give 6.516 million gas for adding H and 29.703
million for adding L. These figures are the deployment gas of the measured five-blockchain provisioning package with its route defaults. Reusing an already configured HL path needs no new
communication deployment. At the sample means, it costs less than adding H
for the first 7 executions and less than adding L for the first 43 executions
(Figure~\ref{fig:cost-crossover}(a)).

Building that HL path from scratch costs 15.168 million gas for three XIR
Gateway--XIR Registry pairs, four remote-bound XIR Adapters, the transition recorder,
and their configuration. The identical destination receiver deployment is
excluded from both alternatives. This newly built path costs less than adding
L through execution 21. Adding H has both a lower initial cost and a lower
recurring cost, so it has no positive crossover with the newly built path
(Figure~\ref{fig:cost-crossover}(b)). Because XIR Gateways and XIR Registries
serve every application on their blockchain, this initial cost applies to
the first path of a deployment. When both expansions are feasible, the
cheaper one is the relevant baseline.

\begin{table}[t]
\caption{\textbf{Resource fee per execution under four price scenarios.}}
\label{tab:fee-scenarios}
\centering\footnotesize
\setlength{\tabcolsep}{3pt}
\begin{tabular}{@{}lrrrr@{}}
\toprule
Path & 1 gwei & 10 gwei & 30 gwei & Rollup \\
\midrule
Direct H & $2.53\times10^{-4}$ & $2.53\times10^{-3}$ & $7.59\times10^{-3}$ & $2.54\times10^{-6}$ \\
H & $4.95\times10^{-4}$ & $4.95\times10^{-3}$ & $1.48\times10^{-2}$ & $4.98\times10^{-6}$ \\
HH & $7.94\times10^{-4}$ & $7.94\times10^{-3}$ & $2.38\times10^{-2}$ & $7.99\times10^{-6}$ \\
HHH & $1.14\times10^{-3}$ & $1.14\times10^{-2}$ & $3.43\times10^{-2}$ & $1.15\times10^{-5}$ \\
HHHH & $1.52\times10^{-3}$ & $1.52\times10^{-2}$ & $4.57\times10^{-2}$ & $1.53\times10^{-5}$ \\
Direct L & $4.72\times10^{-4}$ & $4.72\times10^{-3}$ & $1.42\times10^{-2}$ & $4.74\times10^{-6}$ \\
L & $7.13\times10^{-4}$ & $7.13\times10^{-3}$ & $2.14\times10^{-2}$ & $7.16\times10^{-6}$ \\
HL & $1.16\times10^{-3}$ & $1.16\times10^{-2}$ & $3.48\times10^{-2}$ & $1.17\times10^{-5}$ \\
HLH & $1.68\times10^{-3}$ & $1.68\times10^{-2}$ & $5.05\times10^{-2}$ & $1.69\times10^{-5}$ \\
HHL & $1.53\times10^{-3}$ & $1.53\times10^{-2}$ & $4.60\times10^{-2}$ & $1.54\times10^{-5}$ \\
HHHL & $1.95\times10^{-3}$ & $1.95\times10^{-2}$ & $5.86\times10^{-2}$ & $1.97\times10^{-5}$ \\
HLHL & $2.50\times10^{-3}$ & $2.50\times10^{-2}$ & $7.49\times10^{-2}$ & $2.51\times10^{-5}$ \\
LHLH & $2.50\times10^{-3}$ & $2.50\times10^{-2}$ & $7.51\times10^{-2}$ & $2.52\times10^{-5}$ \\
\bottomrule
\end{tabular}
\par\smallskip\raggedright
All columns are in ETH.
Rollup: L2 gas at 0.01 gwei plus L1 data fees from OP Mainnet block
157,881,361; operator fees are zero at that block. Protocol service fees
are excluded. Values are rounded to three significant digits.
At USD 3,000 per ETH, 0.001 ETH is USD 3.
\end{table}

These comparisons assume that the application can add the missing protocol.
Hyperlane documents a procedure that lets any party deploy it to a new
blockchain~\cite{hyperlanedeploy}. LayerZero applications connect through the
endpoints listed in its official deployment registry~\cite{layerzerodeployments}.
When the application can add neither missing protocol, a direct connection is
unavailable at any execution volume, and composition through an existing path
is the only way to reach the destination.
Beyond the HL example, applying the measured path templates to EVM endpoints identified by public
registry metadata prices candidate paths for 40,431 of the 78,953 reachable
pairs; the supplementary tables list the priced and unpriced pairs.

\begin{figure}[t]
\centering
\includegraphics[width=\columnwidth]{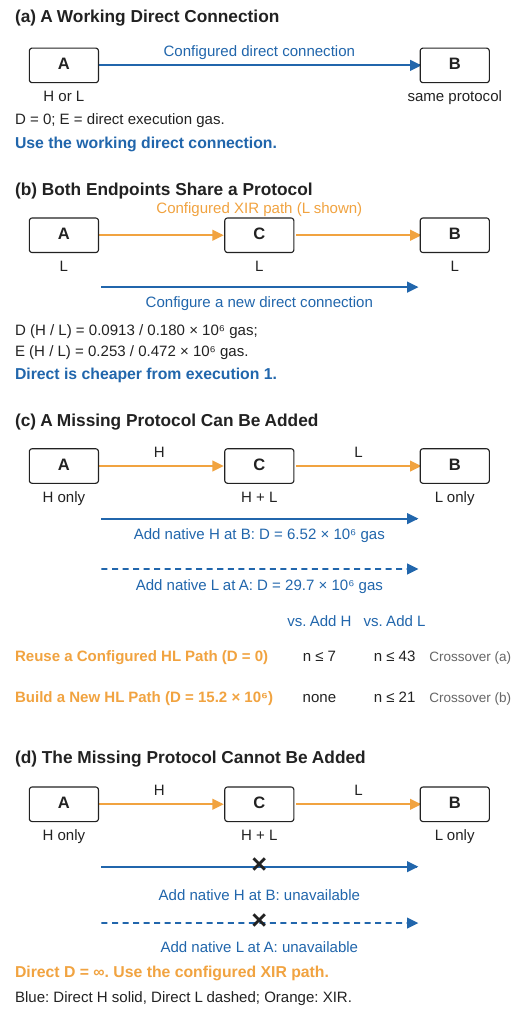}
\caption{\textbf{Four cases for direct delivery and XIR.} XIR paths are orange, direct alternatives blue. Solid blue denotes H and dashed blue denotes L. $D$ is new deployment and configuration gas; $E$ is gas per execution. Panel (b) illustrates L and lists the measured H/L alternatives. Panel (c) uses the HL path and links to the two panels of Figure~\ref{fig:cost-crossover}; the counts are the last integer executions at which XIR is cheaper. none: no positive crossover. H: Hyperlane; L: LayerZero.}
\label{fig:direct-vs-xir-cases}
\end{figure}

\begin{figure}[t]
\centering
\includegraphics[width=\columnwidth]{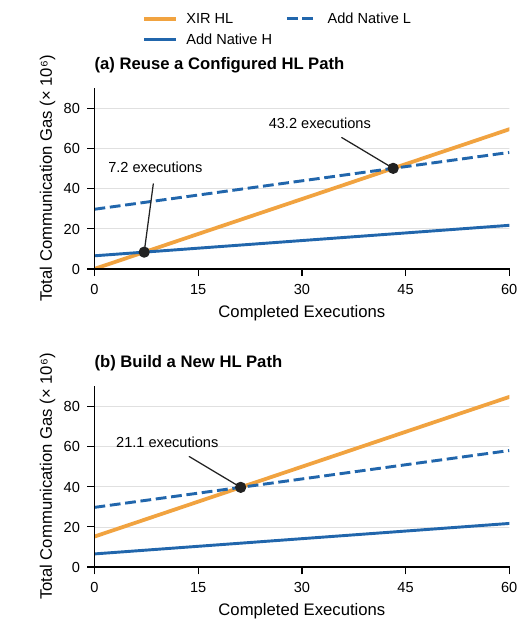}
\caption{\textbf{Cumulative cost of HL and direct delivery.} Lines use $C(n)=D+nE$ with measured deployment gas and sample mean execution gas. Reuse is cheaper through 7 executions against adding H and 43 against adding L; a new HL path is cheaper through 21 against L and has no positive crossover against H. Intersections are continuous thresholds for the measured provisioning package, with identical receiver deployment excluded.}
\label{fig:cost-crossover}
\end{figure}

\subsubsection{Gas Prices, Fees, and Message Value}
\label{sec:fee-scenarios}

Gas counts computation, and a fee converts it into currency. A fee scenario
assigns a gas price to each blockchain, usually quoted in gwei ($10^{-9}$
ETH). The resource fee of a completed execution is $\sum_j g_jp_j$,
where transaction $j$ uses $g_j$ gas at price $p_j$. An OP Stack rollup also
charges an L1 data fee for publishing transaction data on Ethereum and, when
configured, an operator fee\footnote{OP Stack fee specification: \url{https://docs.optimism.io/op-stack/transactions/fees}.}.
When a relayer prepays the destination transaction and is reimbursed, the total counts that fee once. Protocol service charges are outside this resource fee.

Table~\ref{tab:fee-scenarios} prices all 13 paths at uniform gas prices of 1, 10, and 30 gwei and in one rollup scenario. The rollup scenario
uses read-only OP Mainnet fee-oracle estimates at block 157,881,361
(2026-10-07 08:24 UTC). At that block, the oracle reports an Ethereum base
fee of 0.148 gwei. The 1-gwei scenario therefore lies above that base fee and
leaves room for a priority fee, while 10 and 30 gwei represent congested
periods. The rollup scenario adds the L1 data fee to L2 gas at 0.01 gwei,
which exceeds the OP Mainnet base fee of 0.000003 gwei at the same block.
Each path's L1 data fee is averaged over one complete execution at each
payload size, and the operator-fee parameters are zero at that block.

At an ETH price of USD 3,000, the rollup scenario gives resource
costs of about USD 0.035 for HL and USD 0.075 for HLHL. The same paths cost
USD 3.48 and USD 7.49 at 1 gwei, and USD 34.77 and USD 74.88 at 10 gwei.
Under a budget of 0.1\% of the transferred value, HL and HLHL
on the rollup suit transfers above about USD 35 and USD 75. At 1 gwei, the
thresholds rise to about USD 3,477 and USD 7,488, before additional service
charges. At 10 gwei, these thresholds reach about USD 34,770 and
USD 74,880. Under this budget, three- and four-hop paths at the mainnet
price scenarios suit high-value transfers; every measured path in the
rollup scenario costs less than USD 0.076 (LHLH is highest). Each transaction pays the gas price of the blockchain where it
executes, so placing intermediate hops on low-fee blockchains lowers the cost
of a path. Low-value or frequent transfers favor shorter paths, low-fee
execution environments, or an amortized direct connection. For executions
without a monetary transfer, the same comparison uses the application's
willingness to pay per operation.

%% file: sections/7-discussion.tex
\section{Discussion}
\label{sec:discussion}

\subsection{Applying XIR to Cross-Chain Messages}

XIR is useful when an application can reach its destination through existing
protocol connections and can accommodate intermediate execution. Each
additional connection contributes protocol delivery work; each handoff
preserves the operation and checks its accumulated evidence. The application
therefore gains access to a larger reachable set while paying for the selected
path. The graph analysis describes this opportunity, and the prototype
measurements quantify its execution cost for Hyperlane~\cite{hyperlanemailbox}
and LayerZero~\cite{layerzerooapp}.

A deployment requires an XIR Gateway on each participating blockchain and
XIR Adapters for the protocols used along its paths. Adding a protocol involves
mapping its message identifiers, authenticated callbacks, and verifier
configuration into the common representation. The current prototype supplies
these mappings for Hyperlane and LayerZero, including selected EVM-to-Solana
paths. The six-protocol graph provides the broader deployment model.

\subsection{Verification and Delivery Conditions}

Each connection retains its native verification mechanism. XIR links the
accepted results through approved XIR Adapters and checks the complete history at
the destination. The delivery theorem consequently depends on sound protocol
callbacks, correct XIR components and source certification, authentic
configuration, and cryptographic binding. In the implemented policy, every
receipt must resolve to an active profile with a sufficient security level;
the policy hash binds the accompanying policy metadata.

Execution uses the configuration available at each step. A profile change can
therefore stop a path that was eligible when selected. Operational route
selection must account for profile validity, intermediate availability, and
protocol delivery conditions. At the destination, atomic consumption of the
message identifier ensures at-most-once committed execution. Successful
completion additionally depends on the participating blockchains, protocol
services, and forwarding components making progress.

The four guarantees depend on finality at different places. G1 relies on
the root signer, which certifies only a finalized source event (A2).
On the private QBFT blockchains, the root signer waits for one successor block and verifies
that the source receipt block has not been rolled back before signing. G2
relies on A1 at every hop, so it inherits the confirmation setting of each
protocol along the path, including hops that start on an intermediate
blockchain. G3 relies on the registered profiles (A3). G4 relies only on the
destination. The message identifier is fixed when the source record is
created, so if an intermediate rollback causes a hop to be sent again, the
destination rejects the second copy as a duplicate.

Table~\ref{tab:rollback-guarantees} separates the three rollback locations.
For each location, the table names the guarantee a rollback can affect,
the mechanism that preserves it, and, in the last column, the setting that
mechanism depends on.
At the source, a rollback before certification leaves no matching certified
root; after certification, A2 requires the source root to remain final.

\begin{table}[t]
\caption{\textbf{Rollback location and the guarantee it can affect.}}
\label{tab:rollback-guarantees}
\centering\footnotesize
\setlength{\tabcolsep}{2pt}
\renewcommand{\arraystretch}{1.12}
\begin{tabular}{@{}>{\raggedright\arraybackslash}p{.18\columnwidth}>{\raggedright\arraybackslash}p{.12\columnwidth}>{\raggedright\arraybackslash}p{.12\columnwidth}>{\raggedright\arraybackslash}p{.085\columnwidth}>{\raggedright\arraybackslash}p{.18\columnwidth}>{\raggedright\arraybackslash}p{.23\columnwidth}@{}}
\toprule
Rollback location & G1\newline Source auth. & G2\newline Ordered evidence & G3\newline Policy & G4\newline At-most-once & Required setting \\
\midrule
Source & Root signer finality rule & --- & --- & --- & Root signer wait: one successor block on QBFT \\
Inter\-medi\-ate & --- & Per-hop A1 & --- & Fixed message identifier & Protocol confirmations (Table~\ref{tab:experimental-settings}): one LayerZero source confirmation; Hyperlane reorg period 0 \\
Desti\-nation & --- & --- & --- & Entry and effect in one transaction & Destination finality \\
\bottomrule
\end{tabular}
\par\smallskip\raggedright
---: a rollback at this location does not affect the guarantee.

\end{table}

Intermediate blockchains need not finalize simultaneously. Before a protocol
carries an intermediate delivery onward, its verifier applies the configured
source-finality rule. If an intermediate block is removed after its evidence
has been accepted downstream, A1 fails for that hop, and XIR cannot reverse
the transaction already committed on the next blockchain. Each protocol's
confirmation setting must therefore cover the finality of the blockchain it
reads from. Both protocol settings control how long a protocol's verifiers
wait before they treat a source block as final. The LayerZero confirmation
count is the number of source blocks the DVN waits for before it verifies a
message; the Hyperlane reorg period is the number of blocks its validators
wait before they sign a checkpoint. Table~\ref{tab:experimental-settings}
uses one LayerZero source confirmation and a Hyperlane reorg period of zero
on private QBFT blockchains, whose committed blocks are final under the QBFT
fault assumption. Networks with probabilistic finality need more
confirmations, chosen for each source blockchain and protocol.

At the destination, a rollback removes the consumed entry and the application
effect together. A successful redelivery commits one effect in the final
destination state, preserving G4. The required setting at this location is
the destination blockchain's own finality; XIR adds no further configuration.

\subsection{Path Length and Deployment}

The measured graph concentrates reachable pairs at short distances: two-hop
shortest paths account for 75.63\% of reachable pairs, and all finite shortest
paths use at most four hops. This distribution motivates the evaluated path
lengths. Longer paths carry more history, increasing destination verification
linearly and aggregate history transmission quadratically in the current
encoding. Compressing the authenticated history is a natural extension for
applications that require longer paths.

Figures~\ref{fig:cost-stage-cube}--\ref{fig:cost-crossover} place XIR's
cost against the direct alternatives. Per execution, XIR does not undercut a
working direct connection: HL uses 1,159.1 thousand gas against 253.0
thousand for direct Hyperlane and 471.7 thousand for direct LayerZero, and
the extra gas sits in the switch and in Verify and Deliver, both of which
grow with the carried history (Figure~\ref{fig:cost-stage-cube}). XIR's
position is in deployment and reach. Building a new HL path costs 15.168
million gas, more than adding a native Hyperlane stack at the missing
endpoint (6.516 million) and about half of adding a native LayerZero stack
(29.703 million); reusing a configured path adds no deployment
(Figure~\ref{fig:direct-vs-xir-cases}(c)). A configured XIR path is cheaper
than the direct alternative for the first $D/(E_{\mathrm{XIR}}-E_{\mathrm{direct}})$
executions, 43 when LayerZero is missing and 7 when Hyperlane is missing
(Figure~\ref{fig:cost-crossover}(a)). One XIR path
also reaches endpoint pairs that no single protocol connects: in the
measured graph, protocol switches add 32,766 reachable ordered pairs to the
46,187 reachable by same-protocol composition
(Section~\ref{sec:background}). Because XIR Gateways and XIR Registries are
shared by every application on a blockchain, each additional XIR deployment
moves more endpoint pairs into the reuse case; it does not lower the
per-execution cost, which remains bounded by the protocol hops and the
history checks. On rollups, all 13 paths cost less than USD 0.076 per
execution at the measured prices (Table~\ref{tab:fee-scenarios}).

The controlled experiments characterize the deployed protocol stacks on local
Besu QBFT~\cite{besuqbft} blockchains. Public-testnet records demonstrate selected paths
across EVM networks and Solana Devnet. Further deployments can evaluate how
network delay, protocol service availability, and additional XIR Adapter
implementations affect the same verification and execution workflow.

%% file: sections/8-conclusion.tex
\section{Conclusion}
\label{sec:conclusion}

This paper presents XIR, a framework that composes existing cross-chain
protocol connections. XIR uses a verifiable intermediate representation that
binds one application cross-chain message to an ordered record of
authenticated protocol deliveries. XIR comprises XIR Gateways, XIR Adapters,
an XIR Registry, and an XIR Router. On the source blockchain, an XIR Gateway
authenticates the application and creates the record. On each intermediate
blockchain, an XIR Gateway verifies the incoming delivery and extends the
record. On the destination blockchain, an XIR Gateway verifies the complete
history before execution.
Theoretical analysis and evaluation show that XIR preserves source authorization and evidence
order, enforces per-hop policies, and prevents duplicate delivery. With XIR
Gateways and XIR Adapters, composition over the measured 286-blockchain
graph reaches 78,953 ordered pairs (96.86\% of all possible pairs), compared with 14.64\% direct reachability. The open-source
Hyperlane--LayerZero prototype demonstrates executable delivery and measures
the cost of carrying verification history across multi-hop paths. Each additional hop adds about 345,000 gas, so an existing direct connection is cheaper per execution, while XIR serves pairs that share no protocol or exchange few executions. Future
work will extend protocol coverage and reduce history cost.

%% file: sections/submission-declarations.tex
\section*{Funding}

This work was supported in part by the National Key R\&D Program of China
under Grant 2023YFB2704803, in part by the National Natural Science Foundation
of China under Grants 62602642 and U22B2032, and in part by the Beijing Advanced
Innovation Center for Future Blockchain and Privacy Computing under Grant
GJJ-24-025.

\section*{Data availability}

The cross-chain event dataset used for the connectivity analysis is publicly
available on Kaggle (version 4):
\url{https://doi.org/10.34740/kaggle/dsv/19934883}.

\section*{Code availability}

The XIR implementation and experiment scripts are available at
\url{https://github.com/lysrain21/XIR}.